%% file: main.tex
\RequirePackage{xcolor}
\documentclass[journal,twoside,web]{ieeecolor}
\usepackage{generic}

\input{packages}
\newcommand\orcidicon[1]{\href{https://orcid.org/#1}{\includegraphics[scale=0.04]{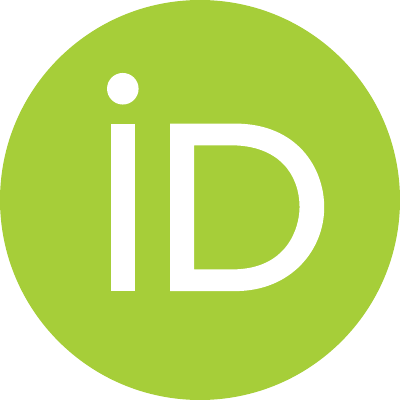}}}

\def\BibTeX{{\rm B\kern-.05em{\sc i\kern-.025em b}\kern-.08em
    T\kern-.1667em\lower.7ex\hbox{E}\kern-.125emX}}
\begin{document}
\bstctlcite{IEEEexample:BSTcontrol}

\title{Network-Induced Strategic Communication in Opinion Dynamics}


\author{Hassan~Munif$^{1,}$\textsuperscript{\orcidicon{0009-0009-8838-243X}}, Anthony~Couthures$^{2,}$\textsuperscript{\orcidicon{0009-0008-6102-9356}}, Vineeth~S.~Varma$^{1,}$\textsuperscript{\orcidicon{0000-0001-8762-2790}}, Samson~Lasaulce$^{1,}$\textsuperscript{\orcidicon{0000-0001-9837-9538}}, and Tamer~Ba\c sar$^{3,}$\textsuperscript{\orcidicon{0000-0003-4406-7875}}
    \thanks{This work was supported by the CNRS MITI project BLESS. (Corresponding author: Hassan~Munif.)}
    \thanks{$^1$~Universit\'e de Lorraine, CNRS, CRAN, F-54000 Nancy, France.}
    \thanks{$^2$~UCLouvain, ICTEAM Institute, Louvain-la-Neuve, 1348, Belgium.}
    \thanks{$^3$~Coordinated Science Laboratory, University of Illinois Urbana-Champaign, Urbana, IL, USA.}
}

\maketitle
\begin{abstract}
    Classical opinion dynamics typically assume a fixed mapping from private opinions to public signals, such as linear exchange, saturated signaling, or discrete public actions. In this paper, we show that these communication mappings can be derived from a strategic communication game played on a weighted influence network. Each agent acts as a receiver estimating its neighbors' states and as a sender broadcasting a public signal to influence its audience. We prove that the network's effect on a sender is summarized by a scalar, network-induced exaggeration factor, and the network game decouples into independent scalar cheap talk problems. The communication rules then emerge as behavioral regimes of one model: aligned incentives recover linear averaging; persuasive senders facing naive receivers produce saturated signaling; and persuasive senders facing Bayesian receivers cannot credibly reveal their opinions, so equilibrium communication becomes an interval quantizer, providing a game-theoretic foundation for continuous-opinion discrete-action (CODA) dynamics. Embedded in repeated opinion updates, the resulting strategic-CODA dynamics preserve opinion clustering and exclude extremism. The model predicts that a speaker exaggerates more when their audience is less influenced by them, and that under strong exaggeration, credible public expression collapses to a binary stance even while private opinions remain continuous.
\end{abstract}

\begin{IEEEkeywords}
    Multi-agent systems, network games, opinion dynamics, strategic communication.
\end{IEEEkeywords}
\section{Introduction}
\label{sec:intro}

\IEEEPARstart{O}{pinion} dynamics study how beliefs evolve when agents repeatedly observe one another over a network. The classical models of DeGroot and Friedkin-Johnsen describe this process through linear social averaging: each agent combines its own opinion with the opinions of its neighbors, possibly with inertia or prejudice \cite{degroot1974,friedkin1990}. This assumption of transparent exchange yields a tractable baseline and, under standard connectivity conditions, predicts consensus or weighted agreement.
However, in social settings, the internal opinion itself is rarely exchanged directly: a voter casts a ballot, a user posts a rating, or a group member endorses an option without fully revealing their private belief. This mismatch reflects \emph{preference falsification} \cite{kuran1995private}, where individuals misrepresent their private views under social pressure. The public act is the only observable quantity; it is typically coarser than the underlying belief and is often chosen strategically with an audience in mind. Once observed, this public signal becomes the social information that drives decisions.

The discrepancy between private and expressed opinions is a core subject in networked opinion dynamics. Recent work shows that quantized messaging and limited attention can generate echo chambers, hysteresis, and lock-in even without hostile interactions \cite{kong2026messaging}. Persistent disagreement can also be produced by stubbornness and structural heterogeneity \cite{acemoglu2013opiniona,ghaderi2014opinion}, bounded confidence \cite{deffuant2000mixing,hegselmann2002opinion,etesami2015game},
antagonistic relations on signed graphs~\cite{altafini2013consensus,liu2017exponential}, and by saturated or sigmoidal interaction laws \cite{bizyaeva2023nonlinear,couthures2025from}.
Continuous-opinion discrete-action (CODA) models explicitly distinguish private opinions from public actions, assuming the former remain continuous while the latter lie in a finite set \cite{martins2008continuous,martins2010importance}; coevolutionary variants couple both layers \cite{aghbolagh2023coevolutionary}. Quantized consensus and networked control models study related effects under finite-rate exchange \cite{chowdhury2016continuous,ceragioli2018consensus,carli2010gossip,etesami2015convergence}. In these works, however, the map from private opinion to public signal is typically fixed a priori.

In this work, we endogenize the opinion-to-signal mapping by deriving it directly from individual incentives. A related effort derives the weighted-median update rule from cognitive-dissonance minimization in a network game \cite{mei2022micro}, but under complete information, with neighbors' states directly observed.
In contrast, our model introduces an asymmetric-information signaling stage: as a receiver, an agent seeks to infer the private opinions of its sources; as a sender, the same agent aims to shift its audience toward its own position. This tension is the classic credibility problem of strategic information transmission (SIT): in the cheap talk model of Crawford and Sobel \cite{crawford1982strategic}, misaligned incentives rule out fully revealing communication and lead to coarse partitions of the state space. The theory has been extended to multiple senders and audiences \cite{krishna2001model,battaglini2002multiple,farrell1989cheap,goltsman2011how} and taken up in the control community for hierarchical communication games, strategic sensing and communication~\cite{akyol2017information,farokhi2017estimation,saritas2017quadratic,kazikli2022signaling}. For scalar cheap talk in which the sender's ideal action is \emph{linear} in the state (the structure our network game induces), \cite{melumad1991communication,gordon2010on} characterized equilibrium coarseness. Closer to our setting, \cite{hagenbach2010strategic} and \cite{galeotti2013strategic} study SIT on networks but in static one-shot formulations that do not embed the resulting communication protocol inside an opinion dynamics feedback loop.
Within opinion dynamics, models have allowed expressed opinions to deviate from private ones: through conformity pressure in \cite{ye2019influence}, and through preference-driven misrepresentation in \cite{buechel2015opinion}. Both, however, are complete-information models: receivers face no inference problem, so communication carries no credibility constraint.
To our knowledge, no prior study derives the opinion-to-signal map of opinion dynamics from strategic incentives, even though coarse, strategically chosen public expression (such as ballots, ratings, and endorsements) is the prevalent mode of social communication. A fundamental question thus remains: \emph{Can linear exchange, saturated expression, and quantized public actions be derived from a single strategic communication model, and how does the network graph determine the resulting communication law?}

In our model, each agent acts as both a receiver and a sender over a fixed influence network. As a receiver, it balances personal inertia against social conformity; as a sender, it broadcasts one public message to influence its audience. We prove that the network topology determines how much the sender must distort its public expression, and show that canonical communication laws emerge as behavioral regimes of this single game. The main contributions of this paper are:

1) We formulate a public-broadcast game on a network, where each agent acts as both a receiver and a sender, and show that audience weights and inertia induce a network-dependent exaggeration factor: the network itself determines how much each individual overstates in public.

2) We prove that, under Bayesian decoding and symmetric priors, the network game decouples into independent scalar cheap talk problems: each agent's communication problem depends on the network only through a scalar.

3) We classify the induced signaling regimes (linear, saturated, and quantized). For the Bayesian regime, we prove that full revelation is impossible, show that interval quantization is optimal, derive closed-form thresholds under a uniform prior for any number of bins, and show that strong exaggeration collapses public discourse to a binary stance.

4) We embed the induced quantizers in the repeated dynamics, obtaining what we denote as the \emph{strategic-CODA dynamics}, and prove forward invariance, monotonicity, existence of interior steady states, and a convergence-switching dichotomy: private opinions converge if and only if the aggregate public signal settles. Strategic quantization thus preserves clustering.

\textbf{Notation.} Bold lower-case and upper-case letters denote vectors and matrices, respectively (e.g., $\bm{v} \in \mathbb{R}^n$, $\bm{W} \in \mathbb{R}^{n \times m}$), with $v_i$ and $w_{ij}$ denoting their respective entries. Column of all ones is $\bm{1}$. Vector inequalities $\bm{u} \leq \bm{v}$ apply component-wise. We let $\operatorname{diag}(\bm{v})$ denote the diagonal matrix with diagonal elements $\bm{v}$. 
Expectations $\mathbb{E}[\cdot]$ are taken with respect to the common prior, using the same notation for random variables and their realizations when no confusion can arise.

\section{Opinion Formation and Induced Dynamics}
\label{sec:problem_formulation}

Consider a population of $N$ agents over a weighted, directed graph $\mathcal{G} = (\mathcal{V}, \mathcal{E}, \bm{W})$ indexed by the node set $\mathcal{V} = \{1, \dots, N\}$. The edge set $\mathcal{E} \subseteq \mathcal{V} \times \mathcal{V}$ defines the communication links, where a directed edge $(j, i) \in \mathcal{E}$ represents an information flow from the source agent $j$ to the receiver agent $i$. Each edge $(j,i)$ is assigned a non-negative weight $w_{ij} \geq 0$ that quantifies the influence of agent $j$ over $i$. These weights form the adjacency matrix $\bm{W} = [w_{ij}]$, which we assume to be row-stochastic ($\bm{W}\bm{1} = \bm{1}$) with no self-loops ($w_{ii} = 0$).

For each agent $i \in \mathcal{V}$, we define two topological neighborhoods: let $\mathcal{N}_i^{\mathrm{S}} = \{j \in \mathcal{V}: w_{ij}>0\}$ denote the set of \emph{sources} (in-neighborhood) representing the agents whose signals are observed by $i$, and $\mathcal{N}_i^{\mathrm{A}} = \{k \in \mathcal{V}: w_{ki}>0\}$ denote the \emph{audience} (out-neighborhood), the agents observing $i$'s signals.

Each agent $i$ holds a private opinion $x_i$ drawn independently from a common prior distribution supported on $[-1,1]$. We denote the probability density function of this prior by $\phi(x)$, which is assumed to be common knowledge. We denote the state space by $\mathcal X\coloneqq[-1,1]^N$.

\begin{standassumption}
    \label{asm:graph}
    The graph $\mathcal{G}$ is \textbf{time-invariant}, and every agent has at least one audience member, i.e., $\mathcal{N}_i^{\mathrm{A}} \neq \emptyset$ for all $i \in \mathcal{V}$. Private opinions $x_i\in[-1,1]$ are mutually independent and drawn from a common prior $\phi(x)$, which is \textbf{symmetric} around zero.
\end{standassumption}

We model the evolution of opinions over this network as a sequence of repeated interactions. Each interaction is governed by a decision protocol consisting of encoding opinions, decoding messages, and updating beliefs.

\begin{figure}[t]
    \centering
    \resizebox{\columnwidth}{!}{
        \begin{tikzpicture}[
                >={Stealth[length=1.8mm,width=1.4mm]},
                font=\footnotesize,
                state/.style={circle, draw, thick, minimum size=6.5mm, inner sep=0pt, fill=white},
                block/.style={rectangle, draw, thick, rounded corners=1pt,
                        minimum height=5mm, minimum width=5.5mm, inner sep=1pt, fill=black!4},
                sourcebox/.style={draw=green!50!black, fill=green!4, dashed, thick,
                        rounded corners=4pt, inner sep=5pt},
                agentbox/.style={draw=red!60!black, fill=red!3, thick,
                        rounded corners=4pt, inner sep=5pt},
                audiencebox/.style={draw=blue!60!black, fill=blue!3, dashed, thick,
                        rounded corners=4pt, inner sep=5pt},
                sarr/.style={->, semithick},
                carr/.style={->, semithick, blue!60!black},
                smalllbl/.style={font=\footnotesize, inner sep=1pt},
                msglbl/.style={font=\footnotesize\itshape, inner sep=1pt,
                        blue!60!black, fill=white},
                boxlbl/.style={font=\footnotesize\bfseries, inner sep=1pt}
            ]

            \node[state] (xj)    at (-1.3, -0.3) {$x_j$};
            \node[block] (sigj)  at (-1.3,  1.0) {$\sigma_j$};

            \node[block] (muj)   at ( 0.6,  1.0) {$\mu_j$};
            \node[block] (zeta)  at ( 1.7,  1.0) {$\zeta_i$};
            \node[state] (zi)    at ( 3.0,  1.0) {$z_i$};

            \node[state, fill=red!8] (xi)   at ( 1.7, -0.3) {$x_i$};
            \node[block]             (sigi) at ( 2.9, -0.3) {$\sigma_i$};

            \node[block] (mui)   at ( 4.8, -0.3) {$\mu_i$};
            \node[block] (zetak) at ( 4.8,  1.0) {$\zeta_k$};

            \begin{scope}[on background layer]
                \node[sourcebox,   opacity=0.35,
                    fit=(xj)(sigj),   xshift=4pt, yshift=-4pt] {};
                \node[audiencebox, opacity=0.35,
                    fit=(mui)(zetak), xshift=4pt, yshift=-4pt] {};
                \node[sourcebox,   fit=(xj)(sigj)]                  (boxj) {};
                \node[agentbox,    fit=(muj)(zeta)(zi)(xi)(sigi)]   (boxi) {};
                \node[audiencebox, fit=(mui)(zetak)]                (boxk) {};
            \end{scope}

            \node[boxlbl, anchor=south, text=green!50!black]
            at (boxj.north) {Source $j \in \mathcal{N}_i^{\mathrm{S}}$};
            \node[boxlbl, anchor=south, text=red!60!black]
            at (boxi.north) {Agent $i$};
            \node[boxlbl, anchor=south, text=blue!60!black]
            at (boxk.north) {Audience $k \in \mathcal{N}_i^{\mathrm{A}}$};

            \draw[sarr] (xj) -- (sigj);
            \draw[carr] (sigj) -- node[msglbl, above]   {$m_j$}
            node[smalllbl, below] {$w_{ij}$} (muj);
            \draw[sarr] (muj) -- node[smalllbl, above] {$\widehat{x}_j$} (zeta);
            \draw[sarr] (zeta) -- (zi);
            \draw[sarr] (xi) -- node[smalllbl, right, pos=0.45] {$\alpha_i$} (zeta);
            \draw[sarr] (xi) -- (sigi);
            \draw[carr] (sigi) -- node[msglbl, above]   {$m_i$}
            node[smalllbl, below] {$w_{ki}$} (mui);
            \draw[sarr] (mui) -- node[right] {$\widehat{x}_i$} (zetak);

            \draw[decorate, decoration={brace, mirror, amplitude=4pt, raise=8pt},
                semithick, gray!55!black]
            ([xshift=-2pt]sigi.south) -- ([xshift=2pt]mui.south);
            \node[font=\footnotesize\itshape, gray!55!black, anchor=north]
            at ($(sigi.south)!0.5!(mui.south) + (0,-0.45)$)
            {$s_i \;:=\; \mu_i \circ \sigma_i$};
        \end{tikzpicture}
    }

    \caption{Communication protocol from the perspective of agent $i$ (red), with representative source $j \in \mathcal{N}_i^{\mathrm{S}}$ (green, dashed) and audience member $k \in \mathcal{N}_i^{\mathrm{A}}$ (blue, dashed). As a receiver (top row), $i$ decodes each message $m_j$ via $\mu_j$ and combines the estimate $\widehat{x}_j$ with its own state $x_i$ through the update map $\zeta_i$. As a sender (bottom row), $i$ encodes $x_i$ via $\sigma_i$ and broadcasts $m_i$ to its audience, decoded by the common $\mu_i$; the effective signaling map is $s_i := \mu_i \circ \sigma_i$.}
    \label{fig:protocol}
\end{figure}
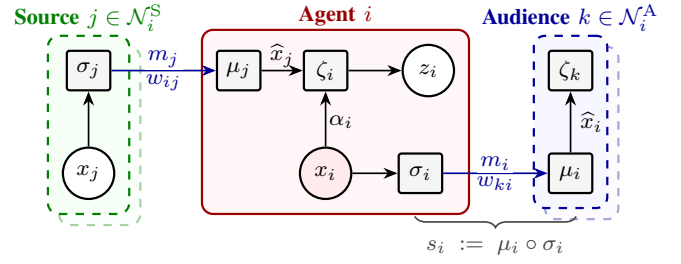

\subsection{Protocol and Agent Objectives}
\label{subsec:protocol_objectives}

To represent strategic public interaction, we consider a public-broadcast protocol consisting of three phases with simultaneous moves within each phase (see Fig.~\ref{fig:protocol}):

\textbf{1) Encoding opinions:} Agent $i$ privately observes its state $x_i$ and applies an \emph{encoding policy} $\sigma_i : [-1, 1] \to \mathcal{M}$ to broadcast a public message $m_i = \sigma_i(x_i)$, where $\mathcal{M}$ is the message space. We assume that the message space is continuous, namely $\mathcal{M} \coloneqq [-1,1]$. Because $\sigma_i$ maps to a single message in $\mathcal{M}$, the agent is restricted from sending targeted messages to individual audience members.

\textbf{2) Decoding messages:} As a receiver, agent $i$ observes the messages from its sources $j \in \mathcal{N}_i^{\mathrm{S}}$. Because private opinions are drawn independently from a common prior, every receiver in $j$'s audience forms an identical posterior estimate of $j$'s state. We therefore define a common decoding policy $\mu_j: \mathcal{M} \to [-1,1]$, such that any receiver listening to sender $j$ computes the estimate $\widehat{x}_j \coloneqq \mu_j(m_j)$.

\textbf{3) Opinion update:} Finally, the agent updates its opinion via an \emph{update policy} $\zeta_i : [-1, 1] \times \mathcal{M}^{|\mathcal{N}_i^{\mathrm{S}}|} \to [-1,1]$ that maps its private opinion and the received messages to an updated opinion $z_i \coloneqq \zeta_i(x_i, \bm{m}_{\mathcal{N}_i^{\mathrm{S}}}) \in [-1,1]$.

Each agent $i \in \mathcal{V}$ minimizes a quadratic cost over $\bm{z} \coloneqq (z_1, \dots, z_N)^\top$ and the profile of private opinions $\bm{x}\coloneqq(x_1, \dots, x_N)^\top$:
\begin{align}
    \label{eq:cost_i}
    c_i(\bm{z},\bm{x}) =\, & \alpha_i\overbrace{(z_i-x_i)^2}^{\text{\small Inertia}} + (1-\alpha_i)\overbrace{\Big(z_i-\sum_{j \in \mathcal{N}_i^{\mathrm{S}}} w_{ij}\,x_j \Big)^2}^{\text{\small Social conformity}} \nonumber \\
                           & +\beta_i \sum_{k \in \mathcal{N}_i^{\mathrm{A}}}\underbrace{{(z_k-x_i)^2}}_{\text{\small Persuasion}}\,,
\end{align}
where $\alpha_i \in [0,1)$ parameterizes the agent's personal \emph{inertia}, and $\beta_i \geq 0$ defines the agent's \emph{persuasion drive}.

The first two terms determine how agent $i$ updates its opinion as a receiver, penalizing deviations of the updated opinion from the private opinion $x_i$ (\emph{inertia}) and from the aggregated private opinions of sources (\emph{social conformity}).

The third term governs the agent's communication strategy as a \textbf{sender}. It imposes a penalty when subsequent updates $z_k$ of the audience deviate from $i$'s private opinion $x_i$. This models the incentive to persuade or manipulate public perception.

\begin{remark}[Influence vs. persuasion]
    \label{rem:weighted_influence}    Influence ($w_{ki}$) dictates how much receiver $k$ listens to source $i$. Persuasion ($\beta_i$) dictates the sender's drive to manipulate the audience. We weight all audience members equally in the persuasion term to isolate network-induced aggregation. $\beta_i = 0$ yields aligned communication; $\beta_i > 0$ introduces strategic conflict.
\end{remark}


The choice of the updated opinion $z_i$ is a local optimization problem: agent $i$ minimizes the conditional expectation of \eqref{eq:cost_i} given its information $(x_i, \bm{m}_{\mathcal{N}_i^{\mathrm{S}}})$. Because the cost is quadratic in $z_i$, this is equivalent to substituting each unknown source state $x_j$ with a decoded point estimate $\mu_j(m_j)$: the posterior mean $\mathbb{E}[x_j \mid m_j]$ under Bayesian decoding, or the face value $m_j$ under naive decoding.

\begin{proposition}[Optimal opinion update]
    \label{prop:optimal_update}
    Given decoded estimates $\mu_j(m_j)$ for all $j\in\mathcal N_i^{\mathrm S}$, the \emph{unique} update $z_i^*$ that minimizes agent $i$'s evaluated cost is:
    \begin{equation}
        \label{eq:optimal_update}
        z_i^{*}=\alpha_i x_i+(1-\alpha_i)\sum_{j \in\mathcal{N}_i^{\mathrm{S}}} w_{ij}\,\mu_j(m_j).
    \end{equation}
\end{proposition}
\begin{proof}
    Agent~$i$ selects $z_i$ to minimize \eqref{eq:cost_i} evaluated at the point estimates $\mu_j(m_j)$. The persuasion term $\beta_i\sum_{k \in \mathcal{N}_i^{\mathrm{A}}}(z_k-x_i)^2$ depends only on $z_k$ (opinion updates of the audience) and is independent of $z_i$. The remaining cost is strictly convex and quadratic in $z_i$ (its second derivative with respect to $z_i$ is $2 > 0$), guaranteeing a unique global minimum. Substituting $x_j$ with $\mu_j(m_j)$, the first-order condition
    yields \eqref{eq:optimal_update}.
\end{proof}
This update rule is DeGroot averaging with a self-weight $\alpha_i$ (inertia), in which the direct observation of neighbors' private states is replaced by their decoded public messages $\mu_j(m_j)$.

The network dynamics depend solely on the composite map from private opinion to decoded public estimate.
For each sender $i$, we define the \textbf{effective signaling function} $s_i : [-1,1] \to [-1,1]$ as $  s_i(x_i) \coloneqq \mu_i(\sigma_i(x_i))$,
and \eqref{eq:optimal_update} shows that all nonlinearity in the closed-loop dynamics is carried by $s_i(\cdot)$.
We then define the macroscopic opinion dynamics:
\begin{equation}
    \label{eq:opinion_update}
    \bm{x}(t+1) \coloneqq \bm{z}^*(t),
\end{equation}
where the optimal update in period $t$ becomes the private opinion entering period $t+1$.

In this framework, the map $s_i$ is the \emph{signaling protocol} through which agents turn internal opinions into public signals. A fully rational treatment would re-derive each protocol every period against the evolving, correlated joint distribution of $\bm{x}(t)$. This dynamic Bayesian game would require tracking a high-dimensional distribution that becomes interdependent after the first interaction. We therefore approximate this process by assuming \emph{myopic agents}:  just as the update \eqref{eq:optimal_update} is a one-step cost minimization rather than a farsighted policy, each agent optimizes its signaling protocol once,  in the stage game, and reuses it in every period. This formulation is standard in strategic networked opinion dynamics \cite{etesami2019influence, forster2016trust}, and it admits a two-time-scale reading: the communication convention $s_i$ forms on a slow social time scale, while opinion updates within that convention occur on the fast time scale captured by \eqref{eq:opinion_update}. This helps to isolate how network topology shapes the geometry of credible signaling. Therefore, we make the following assumption.

\begin{standassumption}[Stationary protocol]
    \label{asm:stationary_comm}
    For each agent $i$, the effective signaling map $s_i:[-1,1]\to[-1,1]$ is fixed at its equilibrium value against the common prior $\phi$ and reused unchanged in every period.
\end{standassumption}

Under Assumption~\ref{asm:stationary_comm}, the closed-loop strategic opinion dynamics reduce to:
\begin{equation}
    \label{eq:closed_loop_dynamics}
    x_i(t+1)
    =
    \alpha_i x_i(t)
    +
    (1-\alpha_i)\sum_{j\in\mathcal{N}_i^{\mathrm{S}}} w_{ij}\,s_j(x_j(t)).
\end{equation}

Using the notation $\bm{D}_{\alpha}\coloneqq \operatorname{diag}(\alpha_1,\dots,\alpha_N)$ and $\bm{s}(\bm{x})\coloneqq (s_1(x_1),\dots,s_N(x_N))^\top$, we write \eqref{eq:closed_loop_dynamics} compactly as
\begin{equation}
    \label{eq:vector_closed_loop}
    \bm{x}(t+1)
    =
    \bm{D}_{\alpha}\bm{x}(t)
    +
    (\bm{I}-\bm{D}_{\alpha})\bm{W}\bm{s}(\bm{x}(t)).
\end{equation}
The macroscopic behavior is therefore determined entirely by the shape of the maps $s_i(\cdot)$ derived in the following sections.

\section{Strategic Network Communication and Network-Induced Exaggeration}
\label{sec:exaggeration}

This section analyzes the formation of the signaling protocols. We model the communication phase as a stage game of incomplete information where agents choose encoding and decoding policies to minimize their expected costs.

\subsection{Strategies and Equilibrium Concept}
\label{subsec:equilibrium}
This framework defines a stage game of incomplete information for the signaling phase. Because only the posterior mean of a sender's private opinion enters the receiver's quadratic update rule \eqref{eq:optimal_update}, we formulate Bayesian consistency only in posterior-mean form (no full belief system is required).

For posterior expectations to be well-defined, we restrict our attention to pure, Borel-measurable policies. A \emph{strategy profile} is a tuple $(\boldsymbol{\sigma}, \boldsymbol{\mu}, \boldsymbol{\zeta})$ defining the policies for all agents:

\noindent\textbf{(1) Encoding:} $\sigma_i(x_i) = m_i$ (how agents signal).

\noindent\textbf{(2) Decoding:} $\mu_j(m_j) = \widehat{x}_j$ (how agents interpret).

\noindent\textbf{(3) Updating:} $\zeta_i(x_i, \mathbf{m}_{\mathcal{N}_i^{\mathrm{S}}}) = z_i$ (how agents update opinion).

We use \emph{Perfect Bayesian Equilibrium} (PBE)~\cite{fudenberg1991perfect} as a solution concept. To formalize this, let $\bm{\sigma} \coloneqq (\sigma_1, \dots, \sigma_N)^\top$ denote the joint encoding profile, $\bm{\zeta} \coloneqq (\zeta_1, \dots, \zeta_N)^\top$ denote the joint update profile, and $\bm{\mu} \coloneqq (\mu_1, \dots, \mu_N)^\top$ denote the joint posterior-mean decoding profile. Since initial opinions are independent, this implies that all audience members of an agent $i$ use the same decoder $\mu_i$ to interpret its public messages.
Prior to the realization of private opinions, agent $i$'s \emph{ex ante} expected cost under a given strategy profile is
$$J_i(\bm{\sigma}, \bm{\mu}, \bm{\zeta}) \coloneqq \mathbb{E}_{\bm{x}}\!\left[c_i\bigl(\bm{z},\bm{x}\bigr)\right]
    =\!\!\int_{\mathcal{X}}c_i\bigl(\bm{z}, \bm{x}\bigr) \prod_{j=1}^N \phi(x_j)\;\mathrm{d}\bm{x},$$

\noindent where updated opinions $\bm{z}$ are implicitly determined by the strategy profile $(\bm{\sigma}, \bm{\mu}, \bm{\zeta})$ for a given state realization $\bm{x}$.
\begin{definition}[Perfect Bayesian equilibrium]
    \label{def:pbe}
    A strategy profile $(\boldsymbol{\sigma}^*, \boldsymbol{\mu}^*, \boldsymbol{\zeta}^*)$ constitutes a PBE if, for every $i \in \mathcal{V}$:
    \begin{enumerate}[wide]
        \item \textbf{Optimal Updating:} $\zeta_i^*$ minimizes the expected cost \eqref{eq:cost_i} given the decoding policies $\boldsymbol{\mu}^*$.
        \item \textbf{Bayes Consistency:} The decoding policy is defined for all messages, $\mu_i^* : \mathcal{M} \to [-1,1]$, and satisfies Bayes' rule almost surely (a.s.) on the equilibrium path:
              \[
                  \mu_i^*(\sigma_i^*(x_i))
                  =
                  \mathbb E[x_i\mid \sigma_i^*(x_i)] \quad\text{a.s.}
              \]
        \item \textbf{Optimal Encoding:} $\sigma_i^*$ minimizes the agent's expected cost given the decoding and updating policies of the audience.
    \end{enumerate}
\end{definition}

We now isolate the object through which the graph topology enters a sender's communication problem.

\subsection{The Exaggeration Factor}
\label{subsec:graph_induced_exaggeration}

Consider agent $i$ during the encoding phase, assuming that the strategies of all other agents are fixed. Because the first two terms in \eqref{eq:cost_i} determine only receiver behavior, the agent's choice of public message $m_i$ affects its cost only through the persuasion term. We can therefore isolate the sender's objective by defining the expected persuasion loss:

\begin{equation}
    \label{eq:sender_loss}
    J^{\mathrm{e}}_i(x_i;\sigma_i) \coloneqq
    \mathbb E\Big[\beta_i\sum_{k\in\mathcal N_i^{\mathrm A}}(z_k^*-x_i)^2\;\Big|\; x_i
        \Big],
\end{equation}
where the expectation is over the states and messages.

For each audience member $k\in\mathcal N_i^{\mathrm A}$, define an effective audience weight, called the \emph{susceptibility weight}:
\begin{equation*}
    \widetilde{w}_{ki}\coloneqq (1-\alpha_k)w_{ki}.
\end{equation*}
This is the effective weight with which receiver $k$'s update uses sender $i$'s decoded signal.

\begin{definition}[Network-induced exaggeration factor]
    \label{def:exagg}
    For each sender $i$, define its \emph{exaggeration factor} by:
    \begin{equation}
        \label{eq:exag_def}
        \xi_i \coloneqq \frac{\sum_{k\in\mathcal N_i^{\mathrm A}}\widetilde{w}_{ki}}    {\sum_{k\in\mathcal N_i^{\mathrm A}}\widetilde{w}_{ki}^2} = \frac{\sum_{k\in\mathcal N_i^{\mathrm A}}(1-\alpha_k)w_{ki}}    {\sum_{k\in\mathcal N_i^{\mathrm A}}(1-\alpha_k)^2w_{ki}^2}.
    \end{equation}
\end{definition}

The exaggeration $\xi_i$ measures the inverse effective leverage of sender $i$'s public signal on its audience. Because audience members exhibit inertia ($\alpha_k$), their updated opinions $z_k$ tend to drift toward their own opinions. To counter this and pull the audience toward its desired state, sender $i$ must overstate its public message. When audience members place only a small effective weight $\widetilde{w}_{ki}$ on $i$'s signal (due to high inertia, competing sources, or weak influence), the required exaggeration factor $\xi_i$ increases. The factor $\xi_i$ is a local audience property rather than a global centrality metric; a large $\xi_i$ indicates that the sender has weak aggregate leverage over its immediate listeners. Figure~\ref{fig:xi_topology_gallery} illustrates this dependence on simple graph families with fixed agent counts, inertia, and weights.

To characterize the limits of this factor, note that because $0 < \widetilde{w}_{ki} \le 1$ for all $k \in \mathcal{N}_i^{\mathrm{A}}$, it follows that $\sum_{k \in \mathcal{N}_i^{\mathrm{A}}} \widetilde{w}_{ki}^2 \le \sum_{k \in \mathcal{N}_i^{\mathrm{A}}} \widetilde{w}_{ki}$, which guarantees $\xi_i \ge 1$. The degenerate case $\xi_i = 1$ holds if and only if every audience member has zero inertia ($\alpha_k = 0$) and listens exclusively to $i$ ($w_{ki} = 1$), yielding a truthful target $y_i = x_i$. Otherwise, $\xi_i > 1$.


\begin{figure}[t]
    \centering
    \includegraphics[width=\columnwidth]{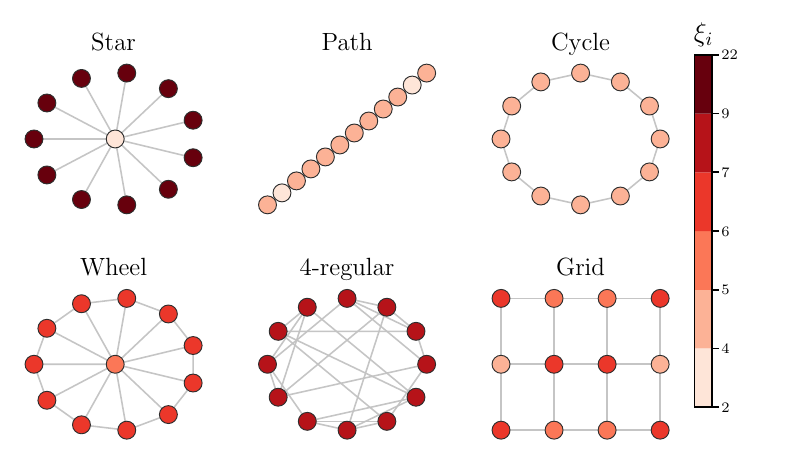}
    \caption{Network-induced exaggeration on example graphs. Each uses $N=12$, common inertia $\alpha_i=0.5$, and uniform weights over neighbors (row-normalized). Different graph topologies induce distinct local profiles of $\xi_i$, showing that the exaggeration factor is a topological quantity rather than a bias imposed externally.}
    \label{fig:xi_topology_gallery}
\end{figure}

\begin{theorem}[Encoding problem reduction]
    \label{thm:cost_reduction}
    Fix a sender $i$ with persuasion drive $\beta_i > 0$. Assuming that its audience applies the optimal update \eqref{eq:optimal_update}, sender $i$'s encoding problem is equivalent to minimizing the reduced objective:
    \begin{equation}
        \label{eq:reduced_cost}
        \mathbb E\left[
            \left(\mu_i(\sigma_i(x_i)) - \xi_i x_i + \delta_i\right)^2
            \right],
    \end{equation}
    where $\delta_i$ is the \emph{network-induced bias} defined as:
    \begin{equation}
        \label{eq:delta_def}
        \delta_i \coloneqq \frac{\sum_{k\in\mathcal N_i^{\mathrm A}} \widetilde{w}_{ki} (1-\alpha_k) \sum_{j\in\mathcal N_k^{\mathrm S}\setminus\{i\}} w_{kj}\mathbb E[\mu_j(m_j)]}{\sum_{k\in\mathcal N_i^{\mathrm A}} \widetilde{w}_{ki}^2}.
    \end{equation}
\end{theorem}

\begin{proof}
    For $k\in\mathcal N_i^{\mathrm A}$, Proposition~\ref{prop:optimal_update} gives
    \[
        z_k^*
        =
        \alpha_kx_k
        +\widetilde{w}_{ki}\mu_i(m_i)
        +(1-\alpha_k)
        \sum_{j\in\mathcal N_k^{\mathrm S}\setminus\{i\}}
        w_{kj}\mu_j(m_j).
    \]

    $\text{Define }
        H_{ki}
        \coloneqq
        \alpha_kx_k
        +(1-\alpha_k)
        \sum_{j\in\mathcal N_k^{\mathrm S}\setminus\{i\}}
        w_{kj}\mu_j(m_j)
        -x_i,
    $
    so that $z_k^*-x_i=\widetilde{w}_{ki}\,\mu_i(m_i)+H_{ki}$.

    By independence and zero-mean private opinions, we have $\mathbb E[x_k \mid x_i] = 0$ and $\mathbb E[\mu_j(m_j) \mid x_i] = \mathbb E[\mu_j(m_j)]$. Thus, taking the conditional expectation gives
    \[
        \mathbb E[H_{ki}\mid x_i] = -x_i + (1-\alpha_k)\!\!\!\!\sum_{j\in\mathcal N_k^{\mathrm S}\setminus\{i\}}\!\!\!\! w_{kj}\mathbb E[\mu_j(m_j)] \eqqcolon -x_i + h_{ki}.
    \]

    Expanding the square in \eqref{eq:sender_loss} and conditioning on $x_i$ yields, up to the multiplicative constant $\beta_i$ and additive terms independent of $m_i$ (such as $\mathbb E[H_{ki}^2\mid x_i]$):

    \begin{align*}
         & \sum_{k\in\mathcal N_i^{\mathrm A}}
        \left(
        \widetilde{w}_{ki}^2\mu_i(m_i)^2
        + 2\widetilde{w}_{ki}\mu_i(m_i)\mathbb E[H_{ki}\mid x_i]
        \right)\nonumber                       \\
         & =
        \sum_{k\in\mathcal N_i^{\mathrm A}}
        \left(
        \widetilde{w}_{ki}^2\mu_i(m_i)^2
        - 2\widetilde{w}_{ki}x_i\mu_i(m_i)
        + 2\widetilde{w}_{ki}h_{ki}\mu_i(m_i)
        \right).
    \end{align*}

    Let $S_1\coloneqq\sum_{k\in\mathcal N_i^{\mathrm A}}\widetilde{w}_{ki}$, $S_2\coloneqq\sum_{k\in\mathcal N_i^{\mathrm A}}\widetilde{w}_{ki}^2$, and let $B_i \coloneqq \sum_{k\in\mathcal N_i^{\mathrm A}} \widetilde{w}_{ki} h_{ki}$. The message-dependent part of the sender's expected cost simplifies to $S_2\mu_i(m_i)^2-2S_1x_i\mu_i(m_i) + 2B_i\mu_i(m_i)$.
    By factoring out $S_2$ and completing the squares for $\mu_i(m_i)$, we obtain
    \begin{equation*}
        S_2\left[\mu_i(m_i) - \left(\frac{S_1}{S_2}x_i - \frac{B_i}{S_2}\right)\right]^2 - S_2\left(\frac{S_1}{S_2}x_i - \frac{B_i}{S_2}\right)^2.
    \end{equation*}

    By \Cref{def:exagg}, $\xi_i = S_1/S_2$ (note $S_2>0$ as $\mathcal{N}_i^{\mathrm{A}} \neq \emptyset$ and $\alpha_k < 1$ imply $\widetilde{w}_{ki}>0$). By expanding $B_i/S_2$, we see that it matches the definition of $\delta_i$ in \eqref{eq:delta_def}. Thus, the conditional expected cost $J^{\mathrm{e}}_i(x_i; \sigma_i)$ depends on the encoded message $m_i = \sigma_i(x_i)$ only through the term $S_2\left(\mu_i(\sigma_i(x_i)) - (\xi_i x_i - \delta_i)\right)^2$. Because the sender's total ex ante persuasion cost is the expectation of $J^{\mathrm{e}}_i(x_i; \sigma_i)$ over the prior of $x_i$, and since $S_2 > 0$ is a constant, taking the outer expectation over $x_i$ yields exactly \eqref{eq:reduced_cost}. Therefore, the sender's ex ante encoding problem is equivalent to minimizing \eqref{eq:reduced_cost}.
\end{proof}

Theorem~\ref{thm:cost_reduction} shows that sender $i$'s public-broadcast problem reduces to a scalar target $y_i(x_i)\coloneqq\xi_i x_i-\delta_i$.

While the gain $\xi_i$ is determined by the graph, the persuasion drive acts only as a switch: $\beta_i$ scales the sender's loss without affecting its minimizer, so every $\beta_i>0$ induces the same target $y_i$, and $\beta_i$ does not enter \eqref{eq:exag_def}. Only the distinction $\beta_i=0$ versus $\beta_i>0$ is behaviorally relevant; a graded persuasion would require an additional trade-off in the sender's cost, such as an intrinsic misreporting penalty, which we deliberately omit to isolate the network-induced mechanism.
The offset $\delta_i$ is a network-induced bias that captures the expected signals that audience members receive from third-party sources. Because its exact computation requires distance-two network information (namely, the in-neighborhoods and weights of audience members' sources), it introduces potential coupling across the entire graph. In the following, we show that this macroscopic network bias vanishes under standard conditions, decoupling the senders' optimization problems.

\begin{proposition}[Vanishing network bias]
    \label{prop:vanishing_bias}
    Under Assumption~\ref{asm:graph}, the network-induced bias vanishes ($\delta_i = 0$ for all $i \in \mathcal{V}$) in either of the following cases:
    \begin{enumerate}[wide]
        \item \textbf{Bayesian Decoding:} In any PBE, regardless of the encoding strategy $\sigma_j$.\label{prop:vanishing_bias:1}
        \item \textbf{Naive Decoding:} If receivers are naive ($\mu_j(m_j) = m_j$) and encoding is origin-symmetric ($\sigma_j(-x) = -\sigma_j(x)$).\label{prop:vanishing_bias:2}
    \end{enumerate}
\end{proposition}

\begin{proof}
    By Assumption \ref{asm:graph}, the prior is symmetric ($\phi(x) = \phi(-x)$ for all $x \in [-1,1]$), which implies $\mathbb{E}[x_i] = 0$.
    For case \ref{prop:vanishing_bias:1}, Bayes' consistency requires that the unconditional expectation of the posterior belief equals the prior mean. Thus, by the Law of Iterated Expectations, $\mathbb{E}[\mu_j(m_j)] = \mathbb{E}[\mathbb{E}[x_j \mid m_j]] = \mathbb{E}[x_j] = 0$.
    For case \ref{prop:vanishing_bias:2}, because the prior is symmetric around zero, any odd encoding function $\sigma_j$ yields an expected broadcast of $\mathbb{E}[\sigma_j(x_j)] = 0$.
    In both cases, substituting $\mathbb{E}[\mu_j(m_j)] = 0$ into \eqref{eq:delta_def} yields $\delta_i = 0$.
\end{proof}

\Cref{prop:vanishing_bias} confirms that, in the Bayesian regime, the reduction to $\delta_i = 0$ is a natural requirement of equilibrium inference, demanding no further restrictive assumptions on the agents' strategy space. Moreover, as shown in \Cref{sec:credible_signaling}, for the uniform prior case, the resulting signaling protocol becomes odd symmetric without any auxiliary assumption.

For the naive regime, the requirement of an odd-symmetric protocol closely mirrors the standard structural assumptions found in other models. For instance, in \cite{bizyaeva2023nonlinear} and \cite{couthures2025from}, interaction functions are defined as odd (e.g., $s(x) = \tanh(x)$ or symmetric saturation) to preserve the neutrality of the origin and the symmetry of the opinion space. Proposition \ref{prop:vanishing_bias} proves that this standard symmetry is exactly what eliminates the network-induced bias and decouples the senders' signaling problems in a strategic setting.

\begin{corollary}[Centered scalar target]
    \label{cor:centered_scalar_target}
    Under the conditions of Proposition \ref{prop:vanishing_bias}, the sender's reduced objective in Theorem \ref{thm:cost_reduction} simplifies to a zero-intercept target:
    \begin{equation}
        \label{eq:centered_target}
        \mathbb E\left[ \left(\mu_i(\sigma_i(x_i))-\xi_i x_i\right)^2 \right].
    \end{equation}
\end{corollary}

\begin{remark}[Public broadcast vs. targeted communication]
    \label{rem:targeted}
    If sender $i$ could target each audience member $k$ separately, the persuasion loss for receiver $k$ would be $(\widetilde{w}_{ki}y_{ki}-x_i)^2$, with targeted optimum $y_{ki}=x_i/\widetilde{w}_{ki}$. Public broadcast forces a single decoded signal, whose least-squares optimum over the audience is $y_i=\xi_i x_i$; thus $\xi_i$ aggregates receiver-specific persuasion incentives into one public-broadcast gain.
\end{remark}

\subsection{Exact Network-to-Scalar Reduction}
\label{subsec:exact_network_reduction}
We show that under Bayesian decoding, the strategic communication game decouples into independent scalar cheap talk games, which characterize the network game's equilibria. Theorem~\ref{thm:cost_reduction} decouples the sender's problem. We extend this decoupling to a global equilibrium equivalence in our main result below.

\begin{theorem}[Exact decoupling into scalar cheap talk]
    \label{thm:exact_reduction}
    A pure strategy profile $(\bm{\sigma}^*, \bm{\mu}^*,\bm{\zeta}^*)$ constitutes a PBE of the stage network game if and only if, for every agent $i\in\mathcal{V}$:
    \begin{enumerate}[wide]
        \item \textbf{Optimal Updating:} The update policy $\zeta_i^*$ is the unique minimum-cost response given by \eqref{eq:optimal_update}.\label{thm:exact_reduction:c1}
        \item \textbf{Bayes Consistency:} The decoding policy $\mu_i^*$ satisfies Bayes' rule almost surely (a.s.) on the equilibrium path:\label{thm:exact_reduction:c2}
              \[
                  \mu_i^*(\sigma_i^*(x_i)) = \mathbb{E}[x_i \mid \sigma_i^*(x_i)] \quad \text{a.s.}
              \]
        \item \textbf{Optimal Encoding:}\label{thm:exact_reduction:c3}
              \begin{itemize}
                  \item If $\beta_i > 0$, the encoding policy $\sigma_i^*$ is optimal for a scalar cheap talk game with state-dependent target $y_i(x_i) = \xi_i x_i$.
                        Specifically, for almost every $x_i \in [-1, 1]$, the encoded message $\sigma_i^*(x_i)$ satisfies:
                        $$ \left(\mu_i^*(\sigma_i^*(x_i)) - \xi_i x_i\right)^2 \le \left(\mu_i^*(m) - \xi_i x_i\right)^2, \quad \forall m \in \mathcal{M}. $$
                  \item If $\beta_i = 0$, the cost is independent of $m_i$ and any encoding, e.g., truthful revelation $\sigma_i^*(x_i) = x_i$, is weakly dominant.
              \end{itemize}
    \end{enumerate}
\end{theorem}

\begin{proof}
    See Appendix~\ref{app:proof_exact_reduction}.
\end{proof}

Theorem \ref{thm:exact_reduction} shows that, despite network coupling, finding a strategic communication equilibrium reduces to solving independent scalar cheap talk games where the network enters only through $\xi_i$.

\begin{remark}[State-dependent bias]
    \label{rem:CS_connection}
    In the classical Crawford-Sobel cheap talk framework \cite{crawford1982strategic}, the sender's loss is typically modeled with a constant bias $b \neq 0$, i.e., $(y - (x+b))^2$. In contrast, our network-derived objective in \eqref{eq:centered_target} can be rewritten as $(y - (x_i + \mathrm{bias}_i(x_i)))^2$, where $\mathrm{bias}_i(x) \coloneqq (\xi_i - 1)x$ acts as a \emph{state-dependent bias}. Because $\xi_i \geq 1$, this bias is proportional to the state itself and points strictly outward from the neutral opinion $x=0$. Consequently, the sender's incentive to manipulate is not a uniform preference for higher or lower actions, but an incentive to \emph{exaggerate}. The more extreme the agent's private opinion, the stronger the incentive to pull the audience further toward the boundaries.
\end{remark}

\section{Classification of Signaling Regimes}
\label{sec:taxonomy}
In this section, we evaluate the decoupled scalar game to classify the communication regimes, by varying two behavioral assumptions (summarized in Table~\ref{tab:taxonomy}):
\begin{enumerate}[wide]
    \item \textbf{Sender Incentives:} \emph{Is the sender's encoding objective aligned with the receiver's inference goal ($\beta_i = 0$), or is it distorted by a persuasion incentive ($\beta_i > 0$)?}
    \item \textbf{Receiver Decoding:} \emph{Do receivers interpret public messages at face value (naive decoding), or do they compute posterior estimates (Bayesian decoding)?}
\end{enumerate}

\subsection{Aligned Incentives and Linear Averaging}
\label{subsec:aligned_incentives}

If the persuasion incentive is absent ($\beta_i=0$), the sender's cost is independent of its own message and every encoding policy is weakly dominant (Theorem~\ref{thm:exact_reduction}); truthful revelation is then the natural convention, as it is the one minimizing every receiver's estimation loss. Selecting this efficient convention, $\sigma_i(x_i)=x_i$ and $\mu_i(m_i)=m_i$, and hence $s_i(x_i)=x_i$.
Substituting this signal function into \eqref{eq:vector_closed_loop} gives
\[
    \bm x(t+1)
    =
    \bm D_\alpha\bm x(t)
    +(\bm I-\bm D_\alpha)\bm W\bm x(t).
\]
This is DeGroot averaging \cite{degroot1974}, with $\bm D_\alpha$ acting as a self-weight; augmenting the cost with a fixed prejudice term recovers a Friedkin-Johnsen-type \cite{friedkin1990} dynamic.

\subsection{Persuasion with Naive Receivers: Saturated Signaling}
\label{subsec:saturated_signaling}
A receiver is naive when it decodes a public message at face value, $\mu_i(m_i)=m_i$, without correcting for senders' encoding strategies.

This regime formalizes \emph{persuasion bias} or \emph{naive learning}, a form of bounded rationality in social networks studied in, e.g., \cite{demarzo2003persuasion, golub2010naive}. Because fully Bayesian inference over network topologies is computationally intractable for human agents, individuals often rely on the heuristic of treating incoming messages as independent and truthful \cite{demarzo2003persuasion}, an assumption extended to strategic network interactions \cite{forster2016trust}.

Under this behavioral assumption, Theorem~\ref{thm:cost_reduction} yields that the unconstrained target of a persuasive sender is $m_i=\xi_i x_i$. Since $\xi_i\ge 1$, the sender possesses a structural incentive to exaggerate outward from the neutral opinion.

If the publicly interpretable signal range is bounded to $[-1,1]$, the optimal constrained signal is given by the clipped saturation function:
\begin{equation}
    \label{eq:saturated_target}
    s_i(x_i)
    =
    \operatorname{sat}(\xi_i x_i-\delta_i)
    \coloneqq
    \min(1,\max(-1,\xi_i x_i-\delta_i)).
\end{equation}
In the origin-symmetric naive convention, Proposition~\ref{prop:vanishing_bias} gives $\delta_i=0$, so this reduces to $\operatorname{sat}(\xi_i x_i)$.

Thus, this strategic-sender/naive-receiver regime recovers the saturated/nonlinear dynamics studied in \cite{bizyaeva2023nonlinear,couthures2025from} as a specific behavioral regime.
While existing literature uses sigmoidal and saturated interaction functions as modeling primitives, we show that this saturation emerges endogenously as the optimal strategy of an agent responding to network-induced persuasion incentives against a naive audience.

\begin{table}[t]
    \caption{Classification of Induced Signaling Regimes}
    \label{tab:taxonomy}
    \centering
    \footnotesize
    \setlength{\tabcolsep}{4pt} 
    \renewcommand{\arraystretch}{1.3} 

    \newcolumntype{Y}[1]{>{\hsize=#1\hsize\raggedright\arraybackslash}X}

    \begin{tabularx}{\columnwidth}{|c|l|Y{1.15}|Y{0.85}|}
        \cline{3-4}
        \multicolumn{2}{c|}{}                                              & \multicolumn{2}{c|}{\textbf{Receiver}}                                                                                                                                                                                   \\ \cline{3-4}
        \multicolumn{2}{c|}{}                                              & \multicolumn{1}{c|}{\textbf{Naive Decoding}} & \multicolumn{1}{c|}{\textbf{Bayesian Decoding}}                                                                                                                           \\ \hline
        \multirow{2}{*}[-1.0ex]{\rotatebox[origin=c]{90}{\textbf{Sender}}} & \textbf{Aligned}                             & $s_i(x)=x$ (linear \cite{degroot1974,friedkin1990})                                                       & $s_i(x)=x$ (linear)                                           \\ \cline{2-4}
                                                                           & \textbf{Persuasive}                          & $s_i(x)=\operatorname{sat}(\xi_i x)$ (nonlinear/saturated \cite{bizyaeva2023nonlinear,couthures2025from}) & $s_i(x)=q_i(x)$ (strategic CODA \cite{martins2008continuous}) \\ \hline
    \end{tabularx}
\end{table}
Figure~\ref{fig:scalar_signaling_maps} gives the corresponding signaling map for persuasive senders. The same network-induced target $\xi x$ yields a clipped public signal under naive decoding and an interval quantizer under Bayesian decoding. The thresholds and reconstruction levels are characterized in Section~\ref{sec:credible_signaling}.

\subsection{Persuasion with Bayesian Receivers}
\label{subsec:credibility_gap}

Under Bayesian decoding, receivers anticipate the sender's preference $\xi_i x_i$. For $\xi_i > 1$, full revelation is not credible: a sender has an incentive to shift the estimate closer to $\xi_i x_i$ (for $x_i\neq0$), mirroring Crawford-Sobel cheap talk \cite{crawford1982strategic,melumad1991communication}. In Section~\ref{sec:credible_signaling}, we show that this network-induced misalignment restricts equilibrium communication to interval quantizers.

\section{Credible Signaling under Bayesian Decoding}
\label{sec:credible_signaling}
We fix a generic sender and suppress the network index to analyze the local communication game.

\begin{definition}[Reduced scalar game]
    \label{def:reduced_game}
    Given a gain $\xi>1$ and a prior $\phi$ on $[-1,1]$, the \emph{reduced scalar game} $\Gamma(\xi,\phi)$ is the cheap talk game in which the sender privately observes $x\sim\phi$ and broadcasts $m=\sigma(x)\in\mathcal M$, receivers apply the decoder $y=\mu(m)$, and the sender's loss function is $g(x,y)\coloneqq(y-\xi x)^2$. An \emph{equilibrium} of $\Gamma(\xi,\phi)$ is a pair $(\sigma,\mu)$ satisfying Bayes-consistent decoding and pointwise-optimal encoding, i.e., conditions \ref{thm:exact_reduction:c2} and \ref{thm:exact_reduction:c3} of Theorem~\ref{thm:exact_reduction}.
\end{definition}

By Theorem~\ref{thm:exact_reduction}, a strategy profile is a PBE of the network game if and only if, for every agent $i$ with $\beta_i>0$, the pair $(\sigma_i,\mu_i)$ is an equilibrium of $\Gamma(\xi_i,\phi)$.

\begin{theorem}[Credibility and interval quantization]
    \label{thm:credibility_and_structure}
    Let the prior density in $\Gamma(\xi,\phi)$ satisfy $\phi(x)>0$ on $[-1,1]$, and fix $\xi>1$. In the reduced game $\Gamma(\xi,\phi)$:
    \begin{enumerate}[wide, label=\roman*), ref=(\roman*)]
        \item \textbf{No full revelation:} No equilibrium satisfies
              $\mu(\sigma(x))=x$ for all $x \in [-1,1]$.
              \label{thm:credibility_and_structure:1}

        \item \textbf{Interval quantization:} In any equilibrium, the effective
              signaling map is an interval quantizer almost surely.
              \label{thm:credibility_and_structure:2}
    \end{enumerate}
\end{theorem}
\begin{proof}
    See Appendix~\ref{app:proof_thm3}.
\end{proof}

Theorem~\ref{thm:credibility_and_structure} shows that every equilibrium of the reduced game is an interval quantizer, allowing us to describe equilibria as \emph{credible codes}.
\begin{definition}[Credible code]
    \label{def:credible_code}
    A \emph{credible code} of $\Gamma(\xi,\phi)$ is the interval-quantizer representation of an equilibrium: an ordered partition of $[-1,1]$ into bins $\{\mathcal B_\ell\}$ together with the posterior-mean reconstruction levels $r_\ell=\mathbb E[x\mid x\in\mathcal B_\ell]$, such that every state in $\mathcal B_\ell$ prefers $r_\ell$ to every other level.
\end{definition}

From a sociological perspective, a speaker cannot credibly convey subtle nuances to a skeptical audience. Instead, public communication is limited to a few distinct stances, each representing a range of private opinions.

Next, we characterize the thresholds and reconstruction levels of credible codes for any codebook cardinality $L$.

\subsection{Equilibrium Geometry under a Uniform Prior}
\label{subsec:equilibrium_geometry}
We now specialize to a uniform prior $x\sim\mathcal U[-1,1]$. Consider an $L$-bin credible code with thresholds
\[
    -1=\tau_0<\tau_1<\cdots<\tau_L=1.
\]
On each bin $\mathcal B_\ell=[\tau_{\ell-1},\tau_\ell)$, with the final bin closed at $1$, the \emph{reconstruction level} (or representative) is
\begin{equation}
    \label{eq:bin_midpoint}
    r_\ell
    =
    \mathbb E[x\mid x\in\mathcal B_\ell]
    =
    \frac{\tau_{\ell-1}+\tau_\ell}{2}.
\end{equation}

\begin{proposition}[Uniform-prior indifference conditions]
    \label{prop:indifference}
    An $L$-bin partition with thresholds $\{\tau_\ell\}_{\ell=0}^{L}$ is an
    equilibrium of the reduced game if and only if:
    \begin{enumerate}[wide, label=\roman*), ref=(\roman*)]
        \item \textbf{Boundary conditions:} $\tau_{0} = -1$ and $\tau_{L} = 1$.\label{prop:indifference:1}
        \item \textbf{Indifference:} For each interior threshold $\ell = 1,\dots, L{-}1$,\label{prop:indifference:2}
              \begin{equation}\label{eq:recurrence}
                  \tau_{\ell+1} - (4\xi - 2)\,\tau_\ell + \tau_{\ell-1} = 0.
              \end{equation}
    \end{enumerate}
\end{proposition}

\begin{proof}
\emph{Necessity ($\Rightarrow$):} In any equilibrium, the state at each interior threshold $\tau_\ell$ must be indifferent between the adjacent bins $\mathcal{B}_\ell$ and $\mathcal{B}_{\ell+1}$. Since the sender's pointwise loss is $g(x,y) = (y - \xi\,x)^2$, this indifference requires $(r_{\ell} - \xi\,\tau_\ell)^2 = (r_{\ell+1} - \xi\,\tau_\ell)^2$. Because the bins are strictly ordered ($r_\ell < r_{\ell+1}$), the terms inside the squares cannot be identical. Thus, $(r_\ell - \xi\tau_\ell) = -(r_{\ell+1} - \xi\tau_\ell)$ and rearranging yields $({r_{\ell}+ r_{\ell+1}})/2 = \xi\tau_{\ell}$. Substituting~\eqref{eq:bin_midpoint} into this relation yields~\eqref{eq:recurrence}, showing that any equilibrium partition must satisfy condition~\ref{prop:indifference:2}.

    \medskip\noindent
    \emph{Sufficiency ($\Leftarrow$):} Assuming \ref{prop:indifference:1} and \ref{prop:indifference:2}, it suffices to show that every state $x$ in bin~$\mathcal{B}_\ell$ weakly prefers the reconstruction level $r_{\ell}$ to every other level $r_k$, with $k\neq\ell$. Because the cross-partial derivative is
    \[
        \frac{\partial^2 g}{\partial x\,\partial y}
        = -2\xi < 0,
    \]
    for any two reconstruction levels $y'>y$, the difference $g(x,y') - g(x,y)$ is strictly decreasing in~$x$. Consequently, if the state at threshold $\tau_\ell$ is indifferent between $r_{\ell}$ and $r_{\ell+1}$, then every $x < \tau_\ell$ strictly prefers $r_{\ell}$ and every $x > \tau_\ell$ strictly prefers $r_{\ell+1}$. By the transitivity of this single-crossing property, this extends to all non-adjacent pairs $(r_{\ell}, r_k)$ with $|k-\ell|>1$.
\end{proof}

\begin{remark}[Generic prior]
    \label{rem:generic_prior}
    For the general prior $\phi$ with cumulative distribution function $\Phi$, an interval bin
    $\mathcal B_\ell=[\tau_{\ell-1},\tau_\ell)$ has reconstruction level
    \[
        r_\ell
        =
        \frac{1}
        {\Phi(\tau_\ell)-\Phi(\tau_{\ell-1})} \int_{\tau_{\ell-1}}^{\tau_\ell}x \phi(x)\,\mathrm{d}x.
    \]
    At each interior threshold, sender indifference requires
    $
        {(r_\ell+r_{\ell+1})}/{2}=\xi\tau_\ell.
    $
\end{remark}

The indifference condition~\eqref{eq:recurrence} constitutes a second-order linear difference equation with constant
coefficients. Its characteristic equation is $\theta^2 - (4\xi - 2)\theta + 1 = 0$.

For $\xi > 1$, the discriminant $16\xi(\xi-1)$ is strictly positive. The two real roots satisfy
$\theta_1 \theta_2 = 1$, allowing us to express them as $\theta_1 = e^{\eta}$ and $\theta_2 = e^{-\eta}$, where
\begin{equation}
    \label{eq:eta_def}
    \eta \coloneqq \operatorname{arccosh}(2\xi - 1) > 0.
\end{equation}

\begin{theorem}[Closed-form credible codes under a uniform prior]\label{thm:partition}
    Fix $\xi > 1$, and let $\eta = \operatorname{arccosh}(2\xi-1)$. Then:
    \begin{enumerate}[wide, label=\roman*), ref=(\roman*)]
        \item \textbf{Existence and uniqueness.} For every integer $L \geq 1$, the reduced game has a unique $L$-bin partition equilibrium.\label{thm:partition:1}
        \item \textbf{Closed-form thresholds.} The equilibrium thresholds are given by\label{thm:partition:2}
              \begin{equation}\label{eq:thresh_closed}
                  \tau_\ell= \frac{\sinh\left(\eta\left(\ell-\frac{L}{2}\right)\right)}{\sinh\left(\eta\frac{L}{2}\right)},
                  \qquad \ell = 0,1,\dots, L.
              \end{equation}
        \item \textbf{Symmetry.} The partition is origin-symmetric: $\tau_\ell = -\tau_{L-\ell}$ for all $\ell$.\label{thm:partition:3}
    \end{enumerate}
\end{theorem}

\begin{proof}
    See Appendix~\ref{app:proof_partition}.
\end{proof}

\Cref{thm:partition} shows that thresholds cluster geometrically near the neutral opinion: only near-neutral agents retain a graded public vocabulary, while all committed opinions on each side pool into a single extreme stance.

\definecolor{myred}{HTML}{b2182b}
\definecolor{myblue}{HTML}{2166ac}
\definecolor{myorange}{HTML}{e08214}
\definecolor{guidegray}{gray}{0.55}
\begin{figure}[t]
    \centering
    \begin{subfigure}[t]{0.48\linewidth}
        \centering
        \begin{tikzpicture}
            \begin{axis}[
                    scale only axis,
                    width=4.2cm,
                    height=3.6cm,
                    axis lines=middle,
                    xlabel={$x$},
                    ylabel={$s(x)=\operatorname{sat}(\xi x)$},
                    xmin=-1.1, xmax=1.1,
                    ymin=-1.1, ymax=1.1,
                    xtick={-1, 1},
                    ytick={-1, 1},
                    xticklabel style={font=\scriptsize},
                    yticklabel style={font=\scriptsize, anchor=east},
                    label style={font=\small},
                    xlabel style={at={(ticklabel* cs:1)}, anchor=south east, yshift=0.05cm, xshift=-0.05cm},
                    ylabel style={at={(ticklabel* cs:1)}, anchor=south,yshift=-0.05cm},
                    clip=false,
                    legend style={draw=black!50, thin, rounded corners=2pt, at={(rel axis cs:1.02,0.05)}, anchor=south east, font=\scriptsize, legend cell align=left}
                ]

                \addplot[very thick, color=myred] coordinates {
                        (-1, -1) (-0.909, -1) (0.909, 1) (1, 1)
                    };
                \addlegendentry{$\xi = 1.1$}

                \addplot[very thick, color=myblue, dash pattern=on 5pt off 2.5pt, line cap=round] coordinates {
                        (-1, -1) (-0.333, -1) (0.333, 1) (1, 1)
                    };
                \addlegendentry{$\xi = 3.0$}

                \node[anchor=north west, font=\scriptsize, inner sep=2pt] at (axis cs:0,0) {0};
            \end{axis}
        \end{tikzpicture}
        \caption{Naive Receivers}
        \label{subfig:naive_saturation}
    \end{subfigure}
    \hfill
    \begin{subfigure}[t]{0.48\linewidth}
        \centering
        \begin{tikzpicture}
            \begin{axis}[
                    scale only axis,
                    width=4.2cm,
                    height=3.6cm,
                    axis lines=middle,
                    xlabel={$x$},
                    ylabel={$s(x)=q(x)$},
                    xmin=-1.1, xmax=1.1,
                    ymin=-1.1, ymax=1.1,
                    xtick={-1, 1},
                    ytick={-1, 1},
                    xticklabel style={font=\scriptsize},
                    yticklabel style={font=\scriptsize, anchor=east},
                    label style={font=\small},
                    xlabel style={at={(ticklabel* cs:1)}, anchor=south east, yshift=0.05cm, xshift=-0.05cm},
                    ylabel style={at={(ticklabel* cs:1)}, anchor=south,yshift=-0.05cm},
                    clip=false
                ]

                \addplot[very thick, color=myblue, dash pattern=on 5pt off 2.5pt, line cap=round] coordinates {
                        (-1.000, -0.550) (-0.101, -0.550) (-0.101, -0.055)
                        (-0.009, -0.055) (-0.009, 0.000) (0.009, 0.000)
                        (0.009, 0.055) (0.101, 0.055) (0.101, 0.550) (1.000, 0.550)
                    };

                \addplot[very thick, color=myred] coordinates {
                        (-1.000, -0.737) (-0.475, -0.737) (-0.475, -0.307)
                        (-0.140, -0.307) (-0.140, 0.000) (0.140, 0.000)
                        (0.140, 0.307) (0.475, 0.307) (0.475, 0.737) (1.000, 0.737)
                    };


                \draw[thick, black] (axis cs:0.140, -0.03) -- (axis cs:0.140, 0.03);
                \draw[thick, black] (axis cs:0.475, -0.03) -- (axis cs:0.475, 0.03);

                \draw[thick, black] (axis cs:-0.015, 0.307) -- (axis cs:0.015, 0.307);

                \node[anchor=north, font=\scriptsize, inner sep=2pt, yshift=-1pt] at (axis cs:0.140, 0) {$\tau_{\ell-\!1}$};
                \node[anchor=north, font=\scriptsize, inner sep=2pt, yshift=-1pt] at (axis cs:0.475, 0) {$\tau_{\ell}$};

                \node[anchor=east, font=\scriptsize, text=black, inner sep=2pt, xshift=-1pt] at (axis cs:0, 0.307) {$r_{\ell}$};

                \draw[dotted, guidegray, thin] (axis cs:0.02, 0.307) -- (axis cs:0.140, 0.307);
                \draw[dotted, guidegray, thin] (axis cs:0.475, 0.05) -- (axis cs:0.475, 0.307);

                \draw[dotted, black, thick] (axis cs:0.3075, 0) -- (axis cs:0.3075, 0.307);
                \draw[thick, black] (axis cs:0.3075, 0.03) -- (axis cs:0.3075, -0.03);
                \node[circle, fill=black, inner sep=1pt] at (axis cs:0.3075, 0.307) {};


                \draw[decorate, decoration={brace, amplitude=3pt, mirror}, thick, black]
                (axis cs:0.142, -0.22) -- (axis cs:0.473, -0.22)
                node[midway, below=2pt, font=\scriptsize] {$\mathcal{B}_\ell$};

                \node[circle, fill=myred, inner sep=1.2pt] at (axis cs:0.140, 0.307) {};
                \node[circle, draw=myred, fill=white, thick, inner sep=1.2pt] at (axis cs:0.475, 0.307) {};
            \end{axis}
        \end{tikzpicture}
        \caption{Bayesian Receivers}
        \label{subfig:bayesian_quantization}
    \end{subfigure}
    \caption{Signaling maps induced by the exaggeration factor $\xi$. \textbf{(a)} Naive receivers take messages at face value, yielding the saturated map $s(x)=\operatorname{sat}(\xi x)$. \textbf{(b)} Bayesian receivers account for the sender's exaggeration incentive, reducing credible communication to an interval quantizer $q(x)$ ($L=5$ bins); annotations mark a generic bin $\mathcal{B}_\ell=[\tau_{\ell-1},\tau_\ell)$ and its reconstruction level $r_\ell$.}
    \label{fig:scalar_signaling_maps}
\end{figure}
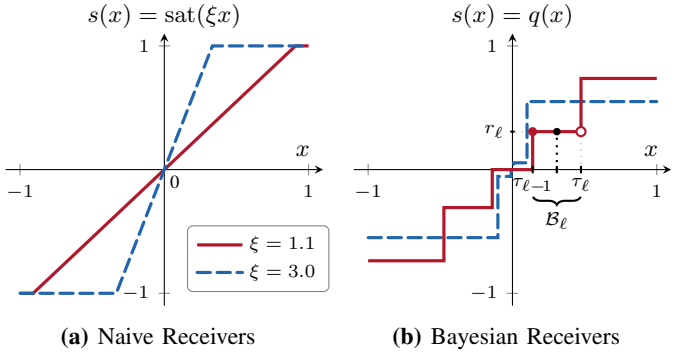

\subsection{The Induced Quantizer Mapping}
\label{subsec:induced_quantizer}
Under the $L_i$-bin partition equilibrium characterized in \Cref{thm:partition}, the communication mechanism of agent $i$ reduces to a scalar quantizer $q_i \coloneqq \mu_i \circ \sigma_i$. Let $\{\tau_{i,\ell}\}_{\ell=0}^{L_i}$ denote the thresholds of agent $i$'s partition, and let
\[
    r_{i,\ell}\coloneqq \frac{\tau_{i,\ell-1}+\tau_{i,\ell}}{2},
    \qquad \ell=1,\dots,L_i,
\]
be the corresponding reconstruction levels. We define the induced quantizer map $q_i:[-1,1]\to[-1,1]$ by
\begin{equation}
    \label{eq:qi_def}
    q_i(x)\coloneqq r_{i,\ell}
    \quad \text{if } x\in(\tau_{i,\ell-1},\tau_{i,\ell}),
    \qquad \ell=1,\dots,L_i,
\end{equation}
extending it to the boundary points by setting $q_i(-1)=r_{i,1}$, $q_i(1)=r_{i,L_i}$, and choosing an arbitrary deterministic tie-breaking rule for the interior thresholds: $q_i(\tau_{i,\ell})\in\{r_{i,\ell},r_{i,\ell+1}\}$, for $\ell=1,\dots,L_i-1$.
This mapping completely encapsulates the outcome of the strategic communication game. To complete the equilibrium construction, the decoder's image over all of $\mathcal{M}$ is restricted to the set of on-path reconstruction levels $\{r_{i,1},\dots,r_{i,L_i}\}$: any off-path message is decoded as one of the on-path levels, so no deviation can generate an estimate outside the codebook.

We now characterize the signaling protocol under a high exaggeration factor ($\xi \gg 1$).

\begin{corollary}[Effective binary public messaging]
    \label{cor:binary_signaling}
    For any credible code of $\Gamma(\xi,\phi)$, every interior threshold satisfies
    \[
        |\tau_\ell|\le \frac{1}{\xi}.
    \]
    Therefore, for $L\ge3$, all non-extreme bins lie in $[-1/\xi,1/\xi]$.
    Under the uniform prior, if $\xi$ satisfies
    \begin{equation}
        \label{eq:collapse_bound}
        \xi\ge \frac{(1+\varepsilon)^2}{4\varepsilon},
        \qquad \varepsilon\in(0,1),
    \end{equation}
    then all interior thresholds lie in $(-\varepsilon,\varepsilon)$.
\end{corollary}

\begin{proof}
    See Appendix~\ref{app:proof_binary_signaling}.

\end{proof}

Strong network-induced exaggeration thus collapses public communication to effectively binary messaging, even as the underlying private opinion remains continuous.

\section{Strategic-CODA Dynamics: Closed-Loop Analysis}
\label{sec:dynamics}

We now characterize the properties of the closed-loop system under the induced signaling maps.
The preceding analysis characterizes the family of $L$-bin partition equilibria; but selecting $L$ is an equilibrium-selection/channel-resolution issue, treated as fixed in the repeated dynamics. This is formalized in the following.


\begin{standassumption}[Stationary strategic quantizer]
    \label{asm:stationary_partition}
    For each agent $i\in\mathcal V$, fix an integer $L_i\ge1$ and select the corresponding $L_i$-bin uniform-prior credible code characterized by Theorem~\ref{thm:partition}. In the repeated model, agent $i$ reuses this selected quantizer $q_i$ in every period.
\end{standassumption}


\begin{figure*}[t!]
    \centering
    \includegraphics[width=0.85\textwidth]{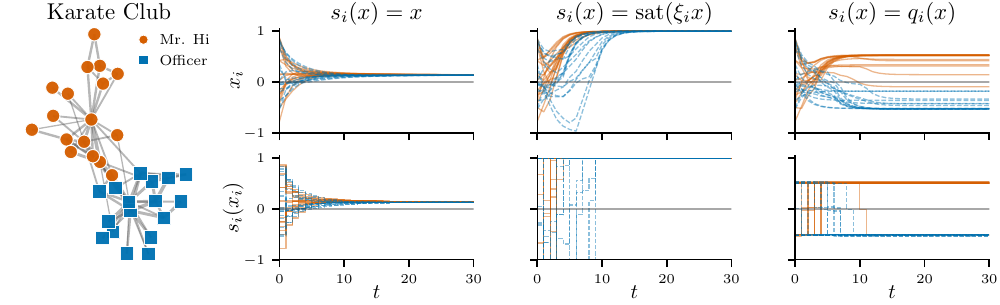}
    \caption{Private opinions $x_i(t)$ (top) and public signals $s_i(x_i(t))$ (bottom) on Zachary's Karate Club network \cite{zachary1977information} ($\alpha_i=0.5$, $x_i(0) \sim \mathcal{U}[-1,1]$, $T=30$, $L=5$). Trajectories are colored by the ``Mr. Hi'' (orange) and ``Officer'' (blue) factions. Varying the signaling protocol $s_i$ yields three distinct configurations: interior consensus (linear), boundary consensus (saturated), or polarized clustering (Bayesian).}
    \label{fig:section7_karate_private_public_signature}
\end{figure*}

\begin{remark}[Selection of the bin count]
    \label{rem:L_selection}
    Theorem~\ref{thm:partition} yields a credible code for \emph{every} $L\ge1$; Assumption~\ref{asm:stationary_partition} selects which one the population coordinates on. This multiplicity is structural: unlike the constant-bias model of \cite{crawford1982strategic}, where misalignment bounds the number of bins \cite{kazikli2022signaling}, the outward bias $(\xi_i-1)x$ vanishes at the neutral opinion, so no cardinality is ruled out \cite{gordon2010on}. We view $L_i$ as the resolution of agent $i$'s public channel (e.g., a rating scale), fixed on the slow time scale together with the convention itself; no subsequent result depends on this selection, since even as $L\to\infty$ the code remains coarse, with outermost interior thresholds at $\pm e^{-\eta}$.
\end{remark}

Substituting $\mu_j(m_j(t))=q_j(x_j(t))$ into \eqref{eq:optimal_update} gives the strategic-CODA dynamics
\begin{equation}
    \label{eq:closed_loop_dynamics_coda}
    x_i(t+1)
    =
    \alpha_i x_i(t)
    +(1-\alpha_i)\sum_{j\in\mathcal N_i^{\mathrm S}}w_{ij}q_j(x_j(t)).
\end{equation}
Let $\bm q(\bm x)\coloneqq (q_1(x_1),\dots,q_N(x_N))^\top$. Then
\begin{equation}
    \label{eq:vector_closed_loop_coda}
    \bm x(t+1)
    =
    \bm D_\alpha\bm x(t)
    +(\bm I-\bm D_\alpha)\bm W\bm q(\bm x(t))
    \eqqcolon
    \bm f(\bm x(t)).
\end{equation}

\begin{lemma}[Quantizer monotonicity]
    \label{lem:qi_properties}
    For every $i\in\mathcal V$, the quantizer $q_i:[-1,1]\to[-1,1]$ is bounded and nondecreasing.
\end{lemma}

\begin{proof}
    Boundedness follows from the fact that reconstruction levels are midpoints of some interval in $[-1,1]$. Monotonicity follows from the ordering of thresholds and reconstruction levels, with a fixed tie-breaking rule at thresholds.
\end{proof}
\begin{proposition}[Forward invariance and order preservation]
    \label{prop:closed_loop_order}
    For any deterministic tie-breaking rule, the map $\bm{f}:\mathcal X\to\mathcal X$ in \eqref{eq:vector_closed_loop_coda} is well-defined, leaves $\mathcal X$ forward-invariant, and is order-preserving.
\end{proposition}

\begin{proof}
    Let $\bm{x}\in\mathcal X$. Since $x_i\in[-1,1]$ and $q_j(x_j)\in[-1,1]$ for every $j$, the quantity $\sum_j w_{ij}q_j(x_j)$ is a convex combination of points in $[-1,1]$, hence also belongs to $[-1,1]$. Therefore, $f_i(\bm{x})= \alpha_i x_i + (1-\alpha_i)\sum_j w_{ij}q_j(x_j) \in [-1,1]$,
    because $0\leq\alpha_i<1$. This proves forward invariance.

    Now let $\bm{u},\bm{v}\in\mathcal X$ satisfy $\bm{u}\le \bm{v}$. By Lemma~\ref{lem:qi_properties}, $q_j(u_j)\le q_j(v_j)$ for every $j$. Since $w_{ij}\ge 0$, it follows that $\sum_j w_{ij}q_j(u_j)\le \sum_j w_{ij}q_j(v_j)$.
    Multiplying by $1-\alpha_i>0$ and adding $\alpha_i u_i\le \alpha_i v_i$ yields $f_i(\bm{u})\le f_i(\bm{v})$ for every $i$. Hence $\bm{f}$ is order-preserving.
\end{proof}

\begin{corollary}[Existence of steady states]
    \label{cor:fixed_point_existence}
    The set of fixed points (steady states) $\mathcal{F}\coloneqq \{\bm{x}\in\mathcal{X}:\bm{f}(\bm{x})=\bm{x}\}$ is nonempty.
\end{corollary}

\begin{proof}
    The partially ordered set $(\mathcal{X},\le)$ endowed with the product order is a complete lattice. By Proposition~\ref{prop:closed_loop_order}, $\bm{f}$ is an order-preserving self-map on $\mathcal{X}$. Tarski's fixed-point theorem \cite{tarski1955lattice} therefore yields a fixed point of $\bm{f}$.
\end{proof}

\begin{theorem}[Fixed-point identity and interiority]
    \label{thm:fixed_points}
    Let $\mathcal F$ be the fixed-point set of \eqref{eq:vector_closed_loop_coda}. Then:
    \begin{enumerate}[wide, label=\roman*), ref=(\roman*)]
        \item $\bm x^*\in\mathcal F$ if and only if\label{thm:fixed_points:1}
              \begin{equation}\label{eq:fixed_point_relation}
                  \bm x^*=\bm W\bm q(\bm x^*).
              \end{equation}
        \item $\mathcal F\subset(-1,1)^N$. Specifically, the extreme consensus states $\bm 1$ and $-\bm 1$ are not fixed points.\label{thm:fixed_points:2}
    \end{enumerate}
\end{theorem}

\begin{proof}
    \ref{thm:fixed_points:1} The fixed-point equation is $\bm{x}^*=\bm{D}_{\alpha}\bm{x}^* +(\bm{I}-\bm{D}_{\alpha})\bm{W}\bm{q}(\bm {x}^*)$.
    Rearranging gives
    $(\bm I-\bm D_\alpha)\,\bm x^*
        =(\bm I-\bm D_\alpha)\bm W\bm q(\bm x^*)$. Since $\alpha_i<1$ for all $i \in \mathcal{V}$, then $(\bm I-\bm D_\alpha)$ is invertible, which proves \eqref{eq:fixed_point_relation}.

    \ref{thm:fixed_points:2} Each $q_j$ takes values strictly inside $[-1,1]$: its minimum reconstruction level is greater than $-1$ and its maximum reconstruction level is less than $1$. Hence there are constants $m>-1$ and $M<1$ such that $m\le q_j(x)\le M$ for all $j$ and all $x\in[-1,1]$. Using \eqref{eq:fixed_point_relation} and row-stochasticity of $\bm W$ gives $m\le x_i^*\le M$ for every $i$.
\end{proof}

\begin{remark}[Impossibility of absolute extremism]
    \label{rem:impossibility_extremism}
    In classical models, networks can settle at the extreme boundaries ($\pm 1$) if the initial states allow. Under Bayesian communication, credibility compresses the outer reconstruction levels strictly inside the opinion space, so polarization and clustering can persist at steady state, but absolute extremism is not possible.
\end{remark}

For fixed signaling functions $\bm{q}(x)$, Theorem~\ref{thm:fixed_points}\ref{thm:fixed_points:1} also shows that the steady states are independent of the inertia profile ($\bm \alpha$), which shapes only the trajectories. We close this section by showing that the trajectories are governed entirely by the public channel. Along a trajectory of \eqref{eq:vector_closed_loop_coda}, let $\bm m(t)\coloneqq \bm q(\bm x(t))$ denote the \emph{public profile} and $\bm c(t)\coloneqq \bm W\bm m(t)$ the aggregate public signal it induces across the network.

\begin{proposition}[Convergence-switching dichotomy]
    \label{prop:dichotomy}
    Along any trajectory of \eqref{eq:vector_closed_loop_coda}:
    \begin{enumerate}[wide, label=\roman*), ref=(\roman*)]
        \item $\bm x(t)$ converges if and only if $\bm c(t)$ is eventually constant. In that case, if $\bm c(t)=\bar{\bm c}$ for all $t\ge T$, the convergence is geometric at each agent's inertia rate,\label{prop:dichotomy:1}
              \begin{equation}
                  \label{eq:conv-switch}
                  x_i(t)-\bar c_i=\alpha_i^{\,t-T}\bigl(x_i(T)-\bar c_i\bigr),
                  \qquad t\ge T,
              \end{equation}
              and the limit $\bar{\bm c}$ satisfies the fixed-point identity \eqref{eq:fixed_point_relation} whenever no coordinate of $\bar{\bm c}$ coincides with a quantizer threshold.
        \item If $\bm f(\bm x(0))\ge\bm x(0)$ or $\bm f(\bm x(0))\le\bm x(0)$, then $\bm m(t)$ is eventually constant, and \ref{prop:dichotomy:1} applies.\label{prop:dichotomy:2}
    \end{enumerate}
\end{proposition}

\begin{proof}
    \ref{prop:dichotomy:1} ($\Leftarrow$) If $\bm c(t)=\bar{\bm c}$ for all $t\ge T$, then \eqref{eq:vector_closed_loop_coda} gives $\bm x(t+1)-\bar{\bm c}=\bm D_\alpha(\bm x(t)-\bar{\bm c})$, and \eqref{eq:conv-switch} follows by iteration;  since $\alpha_i<1$ for all $i$, $\bm x(t)\to\bar{\bm c}$. If no coordinate of $\bar{\bm c}$ is a quantizer threshold, each $q_j$ is constant in a neighborhood of $\bar c_j$, so $\bm q(\bm x(t))=\bm q(\bar{\bm c})$ for all $t$ large enough; hence $\bar{\bm c}=\bm W\bm q(\bar{\bm c})$, which is \eqref{eq:fixed_point_relation}.

    \noindent($\Rightarrow$) If $\bm x(t)\to\bar{\bm x}$, then \eqref{eq:vector_closed_loop_coda} gives $(\bm I-\bm D_\alpha)\bm c(t)=\bm x(t+1)-\bm D_\alpha\bm x(t)\to(\bm I-\bm D_\alpha)\bar{\bm x}$, so $\bm c(t)\to\bar{\bm x}$. A convergent sequence taking values in a finite set is eventually constant.

    \ref{prop:dichotomy:2} Suppose $\bm f(\bm x(0))\ge\bm x(0)$; the decreasing case is symmetric. Since $\bm f$ is order-preserving (Proposition~\ref{prop:closed_loop_order}), induction gives $\bm x(t+1)\ge\bm x(t)$ for all $t$, so each coordinate converges by monotone boundedness. By Lemma~\ref{lem:qi_properties}, each $m_j(t)=q_j(x_j(t))$ is then nondecreasing and takes finitely many values, hence is eventually constant, and so is $\bm c(t)=\bm W\bm m(t)$.
\end{proof}

\Cref{prop:dichotomy} states that private opinions cannot remain in motion on their own: persistent fluctuations of opinions require persistent switching of public stances. Once public discourse settles, every private opinion locks onto the settled public signal geometrically, at its own inertia rate. The convergence of \eqref{eq:vector_closed_loop_coda} thus reduces to whether the public profile settles. Since $\bm f$ is order-preserving but discontinuous, classical convergence theory for monotone dynamical systems does not apply directly.

\section{Numerical Experiments}
\label{sec:simulations}

\begin{figure*}[t!]
    \centering
    \includegraphics[width=0.85\textwidth]{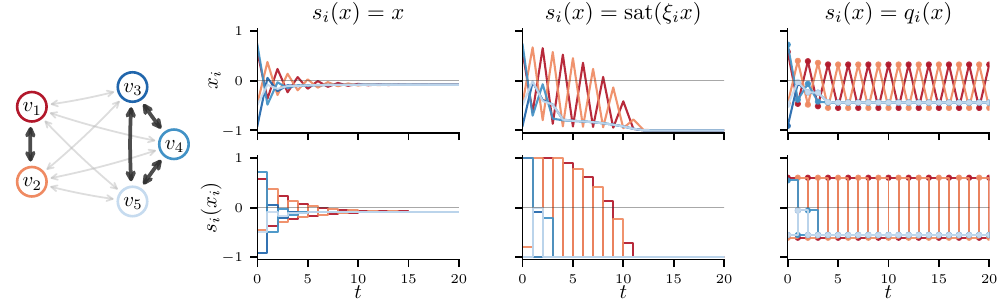}

    \caption{Private opinions and public signals on the two-group five-node network of \Cref{subsec:switching}. The graph, $W$, inertia $\alpha_i = 0.05$, initial condition, horizon $T = 20$, and Bayesian bin count $L = 5$ are fixed; only the signaling map $s_i$ changes. Linear and saturated signaling settle to fixed points within the horizon, whereas Bayesian quantized signaling enters persistent period-2 public switching.}
    \label{fig:section7_private_public_signature}
\end{figure*}

We illustrate the closed-loop dynamics under the three signaling regimes: linear ($s_i(x)=x$), saturated ($s_i(x)=\operatorname{sat}(\xi_i x)$), and quantized ($s_i(x)=q_i(x)$, the strategic-CODA regime). In each experiment, we fix the network, weights, inertia, initial state, and horizon, changing only the signaling map.


\subsection{Consensus Versus Clustering: Karate Club Network}
\label{subsec:karate}

We first consider Zachary's Karate Club network \cite{zachary1977information} as a benchmark. This network captures the friendships among 34 members of a university club that split into two factions: one centered on the instructor (``Mr. Hi'') and the other around the club administrator (``Officer''). We set uniform inertia $\alpha_i=0.5$, and draw initial opinions i.i.d. from $\mathcal{U}[-1,1]$. For the Bayesian regime, we set $L=5$. The remaining parameters are reported in the caption of Fig.~\ref{fig:section7_karate_private_public_signature}.

With aligned communication, continuous averaging eliminates the initial dispersion, and opinions reach an interior consensus. With persuasive senders and naive receivers, saturation amplifies the public signal and, for this initial condition, drives the network to a boundary consensus. The Bayesian quantizer instead maps private opinions into a set of reconstruction levels, and this coarse public channel maintains a clustered, polarized configuration. Consistent with Theorem~\ref{thm:fixed_points}\ref{thm:fixed_points:2}, the preserved clusters are strictly interior, unlike the saturated regime's boundary consensus.

\subsection{Persistent Public Switching on a Two-Group Network}
\label{subsec:switching}

The clustering effect above concerns the configuration; the strategic-CODA regime can also alter the temporal behavior. To isolate this effect, we use a five-node network example with two coupled groups $\{1,2\}$ and $\{3,4,5\}$ (see Fig.~\ref{fig:section7_private_public_signature}): each row of $\bm W$ is normalized (with $w_{ii}=0$), same-group sources split a total weight of $0.80$ uniformly, while cross-group sources split $0.20$ uniformly. All agents have inertia $\alpha_i=0.05$, and the Bayesian regime uses $L=5$ bins. The initial opinion state is $\bm{x}(0)=(0.585,-0.459,-0.911,0.730,-0.487)$ and the time horizon is $T=20$.

Figure~\ref{fig:section7_private_public_signature} separates the private state $x_i(t)$ from the public signal $s_i(x_i(t))$. The linear and saturated regimes settle to fixed points within the horizon, the latter at the boundary. Under Bayesian decoding, the discontinuities of the endogenous quantizer create repeated threshold crossings: two agents alternate between adjacent reconstruction levels and the private states inherit a period-2 orbit. This shows an example of the non-convergent (switching) part of the dichotomy in Proposition~\ref{prop:dichotomy}: the aggregate public signal never settles, so the private opinions cannot converge.

The Bayesian regime is therefore not a discretized version of saturation: the saturated map is continuous after clipping and converges, whereas the credible quantizer can sustain persistent switching of public stances. The convergent branch of the dichotomy is equally visible: re-running the same example with higher inertia ($\alpha_i=0.5$, all else unchanged) allows the public profile to settle after finitely many switches, after which the private opinions converge geometrically at their inertia rates, as in Proposition~\ref{prop:dichotomy}\ref{prop:dichotomy:1}. These experiments suggest inertia stabilizes public discourse: we conjecture that for each topology and quantizer profile there is an inertia threshold above which the public profile settles from every initial state; a complete convergence analysis is the subject of future work.

\section{Conclusion}
\label{sec:conclusion}
This paper develops a game-theoretic micro-foundation for opinion dynamics in which agents observe only the public expressions of their neighbors, rather than their private beliefs. Whereas existing opinion dynamics models fix the opinion-to-signal map in advance, here that map is an equilibrium object, shaped by the influence network itself, and the credibility of public expression becomes the central quantity for both analysis and intervention. A single public-broadcast game unifies the canonical communication laws: depending on sender incentives and on how audiences decode, the same agents produce linear exchange, saturated signaling, or quantized public actions. Because the network game decouples exactly into independent scalar cheap talk problems, the closed loop remains tractable, yielding sharp conclusions that would remain conjectures in a less structured model: network position alone generates public exaggeration, credible expression collapses to an effectively binary vocabulary under weak leverage, and opinion clusters survive strictly inside the opinion space, excluding absolute extremism. The same structure singles out channel resolution and inertia as levers for stabilizing public discourse. Future directions include a complete convergence theory for the strategic-CODA dynamics, in particular conditions under which the public profile settles; the dual roles of the network topology, which simultaneously provides the connectivity that drives agreement and generates the signaling laws; and network design, in which the graph is shaped to achieve a target exaggeration profile and, through it, a desired public vocabulary.

\appendix

\subsection{Proof of Theorem~\ref{thm:exact_reduction}}
\label{app:proof_exact_reduction}
\begin{proof}
    By Definition \ref{def:pbe}, the strategy profile tuple $(\bm{\sigma}^*, \bm{\mu}^*, \bm{\zeta}^*)$ is a PBE if and only if it satisfies Optimal Updating, Bayes Consistency, and Optimal Encoding.

    \emph{Necessity ($\Rightarrow$):} Suppose $(\bm{\sigma}^*, \bm{\mu}^*, \bm{\zeta}^*)$ is a PBE.
    First, Optimal Updating dictates that $\zeta_i^*$ must minimize agent $i$'s expected cost given the decoder profile $\bm{\mu}^*$. By Proposition \ref{prop:optimal_update}, the unique global minimum for this receiver problem is exactly \eqref{eq:optimal_update}, satisfying Condition~\ref{thm:exact_reduction:c1}.
    Second, Bayes Consistency explicitly requires that $\mu_i^*$ satisfies Bayes' rule on the equilibrium path, satisfying Condition~\ref{thm:exact_reduction:c2}.
    Third, Optimal Encoding requires $\sigma_i^*$ to minimize the ex ante expected cost given $\bm{\mu}^*$ and $\bm{\zeta}^*$. If $\beta_i = 0$, the cost does not depend on $m_i$, so any encoding policy is weakly dominant. If $\beta_i > 0$, Theorem \ref{thm:cost_reduction} and Corollary~\ref{cor:centered_scalar_target} establish that the sender's encoding problem is equivalent to minimizing the reduced objective $\mathbb{E}[(\mu_i^*(m_i) - \xi_i x_i)^2]$. Minimizing this expectation ex ante is equivalent to pointwise minimization of the squared error for almost every (a.e.) $x_i$, which is Condition~\ref{thm:exact_reduction:c3}.

    \emph{Sufficiency ($\Leftarrow$):} Suppose the strategy profile satisfies Conditions \ref{thm:exact_reduction:c1}, \ref{thm:exact_reduction:c2}, and \ref{thm:exact_reduction:c3}.
    Condition \ref{thm:exact_reduction:c1} guarantees Optimal Updating because \eqref{eq:optimal_update} is the optimal update according to Proposition \ref{prop:optimal_update}.
    Condition \ref{thm:exact_reduction:c2} directly satisfies the Bayes Consistency requirement of Definition \ref{def:pbe}.
    To verify Optimal Encoding, we confirm that no sender has a profitable unilateral deviation. For any fixed $i$, given the established $\bm{\mu}^*$ and $\bm{\zeta}^*$, agent $i$'s choice of message $m_i$ only affects their own cost through the subsequent updates $z_k^*$ of their audience members $k \in \mathcal{N}_i^{\mathrm{A}}$. By Theorem \ref{thm:cost_reduction}, any variation in agent $i$'s expected cost resulting from a deviation in $\sigma_i$ is equivalently captured by the scalar objective $\mathbb{E}[(\mu_i^*(m_i) - \xi_i x_i)^2]$. Because Condition \ref{thm:exact_reduction:c3} assumes that $\sigma_i^*(x_i)$ minimizes $(\mu_i^*(m) - \xi_i x_i)^2$ over all available messages $m \in \mathcal{M}$ (including off-path messages, for which $\mu_i^*$ is defined by the strategy profile), no alternative encoding rule $\widehat{\sigma}_i$ can strictly decrease the expected cost. If $\beta_i = 0$, the cost is invariant to $m_i$, so a strict improvement is again impossible. Therefore, Optimal Encoding holds, confirming that the strategy profile tuple is a valid PBE.
\end{proof}

\subsection{Proof of \Cref{thm:credibility_and_structure}}
\label{app:proof_thm3}
\begin{proof}
    Let $q\coloneqq\mu\circ\sigma$. By Definition~\ref{def:reduced_game} and
    Theorem~\ref{thm:exact_reduction}, Bayes consistency gives
    $q(x)=\mathbb E[x\mid\sigma(x)]$ a.s. Since $q$ is induced by $\sigma$, the
    law of iterated expectations gives $q(x)=\mathbb E[x\mid q(x)]$ a.s. Sender
    optimality requires $g(x,q(x))\le g(x,\mu(m))$ for all $m\in\mathcal M$, for
    a.e. $x$. In what follows, all pointwise statements are understood after
    redefining $q$ on a null set, without affecting expected payoffs.

    \ref{thm:credibility_and_structure:1}
    Suppose, toward contradiction, that $q(x)=x$ a.s. Let
    $\mathcal T=\{x:q(x)=x\}$ denote the truth-telling set. Since $\phi(x)>0$ on
    $[-1,1]$, we can choose $x\in(0,1/\xi)$ such that $x\in\mathcal T$, $\xi x\in
        \mathcal T$, and sender optimality holds at $x$. Under full revelation, type
    $x$ incurs the strictly positive loss $g(x,x)=(x-\xi x)^2$. By deviating to
    the message $\sigma(\xi x)$, the same type induces the estimate
    $q(\xi x)=\xi x$ and obtains loss zero. This profitable deviation contradicts
    sender optimality. Hence full revelation is impossible for $\xi>1$.

    \ref{thm:credibility_and_structure:2}
    We first show that $q$ is nondecreasing a.s. Take $x_1<x_2$ and write
    $y_j=q(x_j)$. Suppose $y_1>y_2$. Since type $x_1$ can mimic type $x_2$, and
    type $x_2$ can mimic type $x_1$, optimality gives
    $g(x_1,y_1)\le g(x_1,y_2)$ and $g(x_2,y_2)\le g(x_2,y_1)$. For $y_2<y_1$, the
    quadratic loss satisfies the single-crossing relation
    $g(x,y_1)\le g(x,y_2)$ if and only if $x\ge (y_1+y_2)/(2\xi)$. Thus
    $x_1\ge (y_1+y_2)/(2\xi)$ and $x_2\le (y_1+y_2)/(2\xi)$, which contradicts
    $x_1<x_2$. Therefore $q$ is nondecreasing, up to a null-set modification. For any value $y$ in the range of $q$, define the bin
    $\mathcal B_y\coloneqq\{x\in[-1,1]:q(x)=y\}$. Since $q$ is nondecreasing, each
    $\mathcal B_y$ is an interval, possibly degenerate. The nondegenerate bins are
    at most countable (disjoint non-singleton intervals in $[-1,1]$ each
    contain a rational number).

    It remains to rule out sets of singleton bins. Let
    $\mathcal S$ denote the union of singleton bins. On $\mathcal S$, the map $q$
    is injective; hence the posterior-mean condition $q(x)=\mathbb E[x\mid q(x)]$
    implies $q(x)=x$ a.s. on $\mathcal S$. If $\mathcal S$ had positive
    probability, then, since $\phi(x)>0$, the set
    $\mathcal A\coloneqq\mathcal S\cap\{x:q(x)=x\}$ would contain two arbitrarily
    close points on at least one side of the origin. If this occurs on $(0,1)$,
    choose $x_0,x'\in\mathcal A$ with $x_0<x'<\min\{\xi x_0,1\}$. Type $x_0$
    obtains the estimate $q(x_0)=x_0$, but by mimicking type $x'$ it obtains
    $q(x')=x'$, which is strictly closer to the target $\xi x_0$; hence
    $g(x_0,x')<g(x_0,x_0)$, contradicting optimality. A symmetric
    argument covers $(-1,0)$.
    Hence $\mathcal S$ is null.

    Hence, up to a null set, the state space is partitioned into countably
    many nondegenerate intervals. Relabeling them as
    $\{\mathcal B_\ell\}_{\ell\in\mathcal L}$, and letting $r_\ell$ be the common value
    of $q$ on $\mathcal B_\ell$, Bayes consistency gives
    $r_\ell=\mathbb E[x\mid q(x)=r_\ell]=\mathbb E[x\mid x\in\mathcal B_\ell]$.
    Therefore $q$ is constant on the elements of an interval partition,
    with posterior-mean reconstruction levels, almost surely. This is precisely
    an interval quantizer.
\end{proof}

\subsection{Proof of Theorem~\ref{thm:partition}}
\label{app:proof_partition}

\begin{proof}
    We begin with \ref{thm:partition:2}. The general solution to the difference equation~\eqref{eq:recurrence} is $\tau_\ell = C_1 e^{\eta\ell} + C_2 e^{-\eta\ell}$, where the constants $C_1,C_2$ are uniquely determined by the boundary conditions which are: $C_1 + C_2 = -1$, and $C_1 e^{\eta L} + C_2 e^{-\eta L} = 1$. The system admits a unique solution
    \[
        C_1 = \frac{1 + e^{-\eta L}}{2\sinh(\eta L)},
        \qquad
        C_2 = \frac{-(1 + e^{\eta L})}{2\sinh(\eta L)}.
    \]
    Substituting $C_1,C_2$ into the general solution yields:
    \begin{align*}
        \tau_\ell \nonumber
         & = \frac{(1+e^{-\eta L})\,e^{\eta\ell}
                 - (1+e^{\eta L})\,e^{-\eta\ell}}{2\sinh(\eta L)}    \nonumber                              \\
         & = \frac{\sinh(\eta\ell) - \sinh(\eta( L-\ell))}{\sinh(\eta L)}    \label{eq:thresh_closed_step}.
    \end{align*}
    Applying the sum-to-product identity $\sinh(x) - \sinh(y) = 2\cosh(\frac{x+y}{2})\sinh(\frac{x-y}{2})$ to the numerator and the double-angle identity $\sinh(\eta L) = 2\sinh(\eta\frac{L}{2})\cosh(\eta\frac{L}{2})$ to the denominator simplifies to~\eqref{eq:thresh_closed}.

    \ref{thm:partition:1} For existence, it remains to verify that the thresholds characterized in~\eqref{eq:thresh_closed} are strictly
    increasing. We have:
    \begin{equation*}
        \label{eq:thresh_diff}
        \tau_{\ell+1} - \tau_\ell
        =
        \frac{\sinh\!\left(\eta\left(\ell+1-\frac{L}{2}\right)\right)
            - \sinh\!\left(\eta\left(\ell-\frac{L}{2}\right)\right)}
        {\sinh\!\left(\eta\frac{L}{2}\right)}.
    \end{equation*}
    Since $\sinh(x)$ is strictly increasing in $x$, $\eta>0$, and the denominator is strictly positive, this yields
    $\tau_{\ell+1} > \tau_\ell$ for all $\ell = 0,\dots, L-1$. Together with
    Proposition~\ref{prop:indifference}, this establishes that~\eqref{eq:thresh_closed} defines a
    valid $L$-bin equilibrium for every $L \geq 1$.

    Uniqueness follows from part~\ref{thm:partition:2}: the boundary
    conditions determine $C_1, C_2$ uniquely, so at most one $L$-bin partition
    satisfies the indifference conditions.

    \ref{thm:partition:3} Because $\sinh(\cdot)$ is an odd function, evaluating \eqref{eq:thresh_closed} at $L-\ell$ immediately yields $\tau_{L-\ell} = -\tau_\ell$.
\end{proof}

\subsection{Proof of Corollary~\ref{cor:binary_signaling}}
\label{app:proof_binary_signaling}

\begin{proof}
    At an interior threshold, sender indifference gives $\tau_\ell=({r_\ell+r_{\ell+1}})/{2\xi}$.
    Since all reconstruction levels lie in $[-1,1]$, $|\tau_\ell|\le1/\xi$. Because the thresholds are ordered, all bins except the two outermost ones lie within $[-1/\xi,1/\xi]$.
    If the prior has no atom at zero, the probability mass of the non-extreme bins vanishes as $\xi\to\infty$.

    For the uniform-prior bound, symmetry and monotonicity imply that all interior thresholds lie in $(-\varepsilon,\varepsilon)$ if and only if $\tau_{L-1}<\varepsilon$. From \eqref{eq:thresh_closed},
    \[
        \tau_{L-1}
        =
        e^{-\eta}
        \frac{1-e^{-\eta(L-2)}}{1-e^{-\eta L}}
        <
        e^{-\eta},
        \qquad L\ge3.
    \]
    Thus it suffices that $e^{-\eta}\le\varepsilon$, or $\eta\ge\ln(1/\varepsilon)$. Since $\cosh(\eta)=2\xi-1$, this condition is implied by
    \[
        2\xi-1
        \ge
        \cosh(\ln(1/\varepsilon))
        =
        \frac{1+\varepsilon^2}{2\varepsilon},
    \]
    which is equivalent to \eqref{eq:collapse_bound}.
\end{proof}

\bibliographystyle{IEEEtran}

\bibliography{biblio2}
\end{document}

%% file: packages.tex
\usepackage{amssymb,bm}
\usepackage{mathtools}
\usepackage{amsmath} 

\usepackage{amsthm}
\usepackage[noadjust]{cite}

\makeatletter
\@ifundefined{NAT@parse}{}{\let\NAT@parse\undefined}
\makeatother

\providecommand{\citep}[1]{\cite{#1}}
\providecommand{\citet}[1]{\cite{#1}}

\providecommand{\citeauthor}[1]{\cite{#1}}
\providecommand{\citeyear}[1]{\cite{#1}}

\usepackage{tabularx}
\let\labelindent\relax
\usepackage{enumitem}

\usepackage{tikz}
\usepackage{pgfplots}
\pgfplotsset{compat=1.18} 
\usepackage{tkz-graph}

\usetikzlibrary{
  arrows,
  arrows.meta,
  positioning,
  shapes,
  shapes.geometric,
  calc,
  fit,
  backgrounds,
  decorations.pathreplacing
}

\newtheorem{theorem}{Theorem}
\newtheorem{definition}{Definition}
\newtheorem{proposition}{Proposition}

\newtheorem{lemma}{Lemma}
\newtheorem{corollary}{Corollary}
\newtheorem{assumption}{Assumption}
\newtheorem{standassumption}[assumption]{Standing Assumption}
\newtheorem{remark}{Remark}

\usepackage[hidelinks]{hyperref}
\usepackage[capitalize]{cleveref}
\usepackage{tabularx}
\usepackage{multirow}
\usepackage{booktabs}
\usepackage[english]{babel}
\usepackage{subcaption}
\usepackage{graphicx}   